\documentclass[a4paper,UKenglish,cleveref, autoref]{lipics-v2019}

\pdfoutput=1

\title{Reordering Derivatives of Trace Closures of Regular Languages (Full Version)}

\author{Hendrik Maarand}{Department of Software Science, Tallinn University of Technology, Estonia}{hendrik@cs.ioc.ee}{https://orcid.org/0000-0002-1967-4297}{}
\author{Tarmo Uustalu}{School of Computer Science, Reykjavik University, Iceland \and Department of Software Science, Tallinn University of Technology}{tarmo@ru.is}{https://orcid.org/0000-0002-1297-0579}{}

\authorrunning{H. Maarand and T. Uustalu}
\Copyright{Hendrik Maarand and Tarmo Uustalu}%TODO mandatory, please use full first names. LIPIcs license is "CC-BY";  http://creativecommons.org/licenses/by/3.0/

%TODO mandatory; please add comma-separated list of keywords

\category{}%optional, e.g. invited paper

\relatedversion{This is an extended version of the conference paper \cite{confversion}.} 
\supplement{}%optional, e.g. related research data, source code, ... hosted on a repository like zenodo, figshare, GitHub, ...

\funding{Hendrik Maarand: supported by the ERDF funded Estonian national CoE project EXCITE (2014-2020.4.01.15-0018). Both authors: supported by the Estonian Ministry of Education and Research institutional grant no.\ IUT33-13. }%optional, to capture a funding statement, which applies to all authors. Please enter author specific funding statements as fifth argument of the \author macro.

\acknowledgements{We thank Pierre-Louis Curien, Jacques Sakarovitch, Simon Doherty, Georg Struth and Ralf Hinze for inspiring discussions, and our anonymous reviewers for the exceptionally thorough and constructive feedback they gave us.}%optional

\nolinenumbers %uncomment to disable line numbering

\hideLIPIcs  %uncomment to remove references to LIPIcs series (logo, DOI, ...), e.g. when preparing a pre-final version to be uploaded to arXiv or another public repository

\EventEditors{Wan Fokkink and Rob van Glabbeek}
\EventNoEds{2}
\EventLongTitle{30th International Conference on Concurrency Theory (CONCUR 2019)}
\EventShortTitle{CONCUR 2019}
\EventAcronym{CONCUR}
\EventYear{2019}
\EventDate{August 27--30, 2019}
\EventLocation{Amsterdam, the Netherlands}
\EventLogo{}
\SeriesVolume{140}
\ArticleNo{36}
\usepackage{stmaryrd}
\usepackage{proof}
\usepackage[all]{xy}

\usepackage[dvipsnames]{xcolor}

\newcommand{\ssp}{\hspace*{5mm}}

\newcommand{\dg}{\circ}

\newcommand{\eqdf}{=_{\mathrm{df}}}

\newcommand{\Nat}{\mathbb{N}}
\newcommand{\Bool}{\mathbb{B}}

\newcommand{\bff}{\mathsf{ff}}
\newcommand{\btt}{\mathsf{tt}}

\renewcommand{\P}{\mathcal{P}}

\newcommand{\Mf}{\mathcal{M}_\mathrm{f}}

\newcommand{\RE}{\mathsf{RE}}

\newcommand{\sem}[1]{\llbracket #1 \rrbracket}

\newcommand{\shuffle}{\mathrel{{\sqcup} \hspace{-0.18em} {\sqcup}}}

\newcommand{\one}{\mathbf{1}}
\newcommand{\eps}{\varepsilon}

\newcommand{\imp}{\Longrightarrow}
\newcommand{\bwdimp}{\Longleftarrow}
\newcommand{\con}{\wedge}
\newcommand{\dis}{\vee}

\newcommand{\dn}{{\downarrow}}

\newcommand{\N}{N}
\newcommand{\R}{R}

\newcommand{\squig}{\rightsquigarrow}
\newcommand{\RR}{\mathbf{R}}

\newcommand{\lsim}{\mathrel{{\sim}{\lhd}}}
\newcommand{\rsim}{\mathrel{{\rhd}{\sim}}}

\newcommand{\Lex}{\mathrm{Lex}}

\begin{document}

\maketitle

\begin{abstract}
  We provide syntactic derivative-like operations, defined by
  recursion on regular expressions, in the styles of both Brzozowski
  and Antimirov, for trace closures of regular languages. Just as the
  Brzozowski and Antimirov derivative operations for regular
  languages, these syntactic reordering derivative operations yield
  deterministic and nondeterministic automata respectively. But trace
  closures of regular languages are in general not regular, hence
  these automata cannot generally be finite.  Still, as we show, for
  star-connected expressions, the Antimirov and Brzozowski automata,
  suitably quotiented, are finite. We also define a refined version of
  the Antimirov reordering derivative operation where
  parts-of-derivatives (states of the automaton) are nonempty lists of
  regular expressions rather than single regular expressions. We
  define the uniform scattering rank of a language and show that, for
  a regexp whose language has finite uniform scattering rank, the
  truncation of the (generally infinite) refined Antimirov automaton,
  obtained by removing long states, is finite without any quotienting,
  but still accepts the trace closure. We also show that
  star-connected languages have finite uniform scattering rank.
\end{abstract}

\section{Introduction}

Traces were introduced to concurrency theory by Mazurkiewicz
\cite{Mazurkiewicz77,Mazurkiewicz95} as an alternative to words. A
word can be seen as a linear order that is labelled with letters of
the alphabet.  Intuitively, the main idea of traces is that the linear
order, corresponding to sequentiality, is replaced with a partial
order. Sets of words (or word languages) can be used to describe the
behaviour of concurrent systems. Similarly, sets of traces (or trace
languages) can also be used for this purpose. The difference is that
descriptions in terms of traces do not distinguish between different
linear extensions (words) of the same partial order (trace)---they are
considered equivalent. Different linear extensions of the same partial
order can be seen as different observations of the same behaviour.

Given a word language $L$ and a letter $a$, the derivative of $L$
along $a$ is the language consisting of all the words $v$ such that
$av$ belongs to $L$. An essential difference between words and traces
is that a nonempty word (a linear order) has its first letter as the
unique minimal element, but a nonempty trace (a partial order) may
have several minimal elements. A trace from a trace language can be
derived along any of its minimal letters. Clearly, a minimal letter of
a trace need not be the first letter of a word representing this
trace.

It is well-known that the derivative of a regular word language along
a letter is again regular. Brzozowski \cite{Brzozowski64} showed that
a regexp for it can be computed from a regexp for the given language,
and Antimirov \cite{Antimirov96} then further optimized this result.
We show that these syntactic derivative operations generalize to trace
closures (i.e., closures under equivalence) of regular word languages
in the form of syntactic \emph{reordering derivative} operations.

The syntactic derivative operations for regular word languages provide
ways to construct automata from a regexp. The Brzozowski derivative
operation is a function on regexps while the Antimirov derivative
operation is a relation. Accordingly, they yield deterministic and
nondeterministic automata. The set of Brzozowski derivatives of a
regexp (modulo appropriate equations) and the set of Antimirov
parts-of-derivatives are finite, hence so are the resulting
automata. Our generalizations to trace closures of regular languages
similarly give deterministic and nondeterministic automata, but these
cannot be finite in general. Still, as we show, for a star-connected
expression, the Antimirov and Brzozowski automata, suitably
quotiented, are finite. We also develop a finer version of the
Antimirov reordering derivative, where parts-of-derivatives are
nonempty lists of regexps rather than single regexps, and we show that
the set of expressions that can appear in these lists for a given
initial regexp is finite.  We introduce a new notion of \emph{uniform
  scattering rank} of a language (a variant of Hashiguchi's scattering
rank \cite{Hashiguchi91}) and show that, for a regexp whose language
has finite uniform rank, a truncation of the refined reordering
Antimirov automaton accepts its trace closure despite the removed
states, and is finite, without any quotienting.

This is an extension of the conference paper \cite{confversion} with
proofs of the most important propositions and background material on
classical language derivatives and trace closures of regular
languages.

%\section{Word and Trace Languages}

\section{Preliminaries on Word Languages} 

An \emph{alphabet} $\Sigma$ is a finite set (of \emph{letters}). A
\emph{word} over $\Sigma$ is a finite sequence of letters. The set
$\Sigma^*$ of all words over $\Sigma$ is the free monoid on $\Sigma$
with the empty word $\eps$ as the unit and concatenation of words
(denoted by ${\cdot}$ that can be omitted) as the multiplication.  By
$\pi_X(u)$ we mean the projection of a word $u$ to a subalphabet
$X \subseteq \Sigma$, i.e., $\pi_X(u)$ discards from $u$ all letters
which are not in $X$.  We write $| u |$ for the length of a word $u$
and also $| X |$ for the size of a subalphabet $X$. By $| u |_a$ we
mean $| \pi_a (u) |$, i.e., the number of occurrences of $a$ in $u$.
By $\Sigma(u)$ we denote the set of letters that appear in $u$.

A \emph{(word) language} is a subset of $\Sigma^*$. The empty word and
concatenation of words lift to word languages via
$\one \eqdf \{ \eps \}$ and
$L \cdot L' \eqdf \{ uv \mid u \in L \con v \in L' \}$.

\subsection{Regular Languages}

The set $\RE$ of \emph{regular expressions} (in short, \emph{regexps})
over $\Sigma$ is given by the grammar
$ E, F ::= a \mid 0 \mid E + F \mid 1 \mid EF \mid E^*$ where $a$
ranges over $\Sigma$.

The \emph{word-language semantics} of regular expressions is given by
a function $\sem{\_} : \RE \to \P \Sigma^*$ defined recursively by
\[
\small
\begin{array}{rcl@{\qquad}rcl}
\sem{a}      & \eqdf & \{ a \} 
& \sem{1}      & \eqdf & \one \\
\sem{0}      & \eqdf & \emptyset 
& \sem{E F}   & \eqdf & \sem{E} \cdot \sem{F} \\
\sem{E + F} & \eqdf & \sem{E} \cup \sem{F} 
& \sem{E^*}    & \eqdf & \mu X.\, \one \cup \sem{E} \cdot X % \\
% \sem{1}      & \eqdf & \one \\
% \sem{E F}   & \eqdf & \sem{E} \cdot \sem{F} \\
% \sem{E^*}    & \eqdf & \mu X.\, \one \cup \sem{E} \cdot X 
%    =  \bigcup\limits_{n \in \mathbb{N}} \sem{E}^n
% %      \{ u_1\ldots u_n \mid n \in \mathbb{N} \con u_1, \ldots, u_n \in \sem{E} \}
\end{array}
\]

A word language $L$ is said to be \emph{regular} (or \emph{rational})
if $L = \sem{E}$ for some regexp $E$.
Kleene algebras are defined by an equational theory.  
It was shown by Kozen~\cite{Kozen94} that the set
$\{\sem{E} \mid E \in \RE\}$ of all regular languages together with
the language operations $\emptyset$, $\cup$, $\one$, ${\cdot}$,
$(\_)^*$ is the free Kleene algebra on $\Sigma$.  An important
property for us is that $E \doteq F$ iff $\sem{E} = \sem{F}$ where
$\doteq$ refers to valid equations in the Kleene algebra theory.

Kleene's theorem~\cite{Kleene56} says that a word language is rational
iff it is \emph{recognizable}, i.e., accepted by a finite
deterministic automaton (acceptance by a finite nondeterministic
automaton is an equivalent condition because of
determinizability~\cite{RabinS59}).

\subsection{Derivatives of a Language} \label{sec:derivatives-of-a-language}

A word language $L$ is said to be \emph{nullable} $L \dn$, if
$\eps \in L$.  The \emph{derivative} (or \emph{left
  quotient})\footnote{We use the word `derivative' both for languages
  and expressions, reserving the word `quotient' for quotients of sets
  by equivalence relations.} of $L$ along a word $u$ is defined by
$D_u L \eqdf \{v \mid uv \in L\}$.
For any $L$, we have $D_\eps L = L$ as well as
$D_{uv} L = D_v (D_u L)$ for any $u, v \in \Sigma^*$, i.e., the
operation $D : \P \Sigma^* \times \Sigma^* \to \P \Sigma^*$ is a right
action of $\Sigma^*$ on $\P\Sigma^*$.
%
%For any language $L$, 
We also have
$L = \{ \eps \mid L \dn \} \cup \bigcup \{\{a\} \cdot D_a L \mid a \in
\Sigma \}$,
and for any $u \in \Sigma^*$, we have $u \in L$ iff $(D_u L) \dn$.

Derivatives of regular languages are regular. A remarkable fact is
that they can be computed syntactically, on the level of regular
expressions. There are two constructions for this, due to Brzozowski
\cite{Brzozowski64} and Antimirov \cite{Antimirov96}. We review
these in the next two subsections.
% ?????
The Brzozowski and Antimirov derivative operations yield deterministic
resp.\ nondeterministic automata accepting the language of a regular
expression $E$. The Antimirov automaton is finite. The Brzozowski
automaton becomes finite when quotiented by associativity,
commutativity and idempotence for $+$. Identified up to the Kleene
algebra theory, the states of the Brzozowski automaton correspond to
the derivatives of the language $\sem{E}$. Regular languages can be
characterized as languages with finitely many derivatives.

\subsection{Brzozowski Derivative}

Nullability and derivative are semantic notions, defined about
languages. However, Brzozowski~\cite{Brzozowski64} noticed that for
regular languages, one can compute nullability and the derivatives
syntactically, on the level of regular expressions.

\begin{definition} The \emph{syntactic nullability} and the
  \emph{Brzozowski derivative} of a regexp are given by functions
  ${\dn} : \RE \to \Bool$, $D : \RE \times \Sigma \to \RE$ and
  $D : \RE \times \Sigma^* \to RE$ defined recursively by
\[
\begin{array}{rclrcl}
b \dn & \eqdf & \bff & 
D_a b & \eqdf & \mathsf{if~} a = b \mathsf{~then~} 1 \mathsf{~else~} 0 \\
0 \dn & \eqdf & \bff &
D_a 0 & \eqdf & 0 \\
(E + F) \dn & \eqdf & E \dn \vee F \dn &
D_a (E + F) & \eqdf & D_a E + D_a F \\
1 \dn & \eqdf & \btt &
D_a 1 & \eqdf & 0 \\
(E F) \dn & \eqdf & E \dn \wedge F \dn &
D_a (E F) & \eqdf & \mathsf{if~} E \dn \mathsf{~then~} (D_a E) F + D_a F \mathsf{~else~} (D_a E) F \\  
(E^*) \dn & \eqdf & \btt &
D_a (E^*) & \eqdf & (D_a E) E^* \\[2ex]
& & & 
D_\eps E & \eqdf & E \\
& & & 
D_{ua} E & \eqdf & D_a (D_u E) 
\end{array}
\]
\end{definition}

\begin{proposition} For any $E$,
\begin{enumerate}
\item $\sem{E} \dn = E \dn$;
\item for any $a \in \Sigma$, $D_a \sem{E} = \sem{D_a E}$;
\item for any $u \in \Sigma^*$, $D_u \sem{E} = \sem{D_u E}$.
\end{enumerate}
\end{proposition}
 
\begin{corollary} For any $E$,
\begin{enumerate}
\item $\sem{E} = \{\eps \mid E \dn\} \cup \bigcup \{ \{a\} \cdot \sem{D_a E} 
  \mid a \in \Sigma\}$;
\item for any $a \in \Sigma$, $v \in \Sigma^*$, $av \in \sem{E}$ iff $v \in \sem{D_a E}$;
\item for any $u, v \in \Sigma^*$, $uv \in \sem{E}$ iff $v \in \sem{D_u E}$;
\item for any $u \in \Sigma^*$, $u \in \sem{E}$ iff $(D_u E) \dn$.
\end{enumerate}
\end{corollary}

The Brzozowski derivative operation gives a method for turning a
regular expression into a deterministic automaton. For a regexp $E$,
the set of states is $Q^E = \{D_u E \mid u \in \Sigma^*\}$, the
initial state is $q^E_0 = E$, the final states are
$F^E = \{E' \in Q^E \mid E' \dn \}$ and the transition function
$\delta^E$ is defined by $D$ restricted to $Q^E$.

This automaton is generally not finite, but its quotient by a suitable
syntactically defined equivalence relation on the state set is finite,
as we will see in the next subsection.

\subsection{Antimirov Derivative}

Antimirov \cite{Antimirov96} optimized Brzozowski's construction
essentially constructing a nondeterministic finite automaton (NFA)
instead of a DFA, with a smaller number of states and, crucially,
without having to identify states up to equations.

Antimirov's syntactic derivative operation is a multivalued function,
in other words, a relation. Antimirov spoke of ``partial derivatives''
or ``linear factors'', we prefer to use the term
``parts-of-derivatives''.

\begin{definition} The \emph{Antimirov parts-of-derivatives} of a regular
  expression along a letter or a word are given the relations
  ${\to} \subseteq \RE \times \Sigma \times \RE$ and
  ${\to^*} \subseteq \RE \times \Sigma^* \times \RE$ defined
inductively by
\[
\small
\begin{array}{l}
\infer{a \to (a, 1)}{
}
\quad
\infer{E + F \to (a, E')}{
  E \to (a, E')
}
\quad
\infer{E + F \to (a, F')}{
  F \to (a, F')
}
\\[2ex]
\infer{E F \to (a, E' F)}{
  E \to (a, E')
}
\quad
\infer{E F \to (a, F')}{
  E \dn
  &
  F \to (a, F')
}
\quad
\infer{E^* \to (a, E' E^*)}{
  E \to (a, E')
} \\[2ex]
\infer{E \to^* (\eps, E)}{
}
\quad
\infer{E \to^* (ua, E'')}{
  E \to^* (u, E') 
  &
  E' \to (a, E'')
}
\end{array}
\]   
\end{definition}

The Antimirov parts-of-derivatives compute the semantic derivative
collectively.\footnote{If we took languages to be multisets of words
  (i.e., introduced the notion of a word occurring in a language some
  number of times ) and adopted the obvious multisets-of-words semantics of
  regular expressions, the Antimirov parts-of-derivatives would
  compute regular expressions for a partition of the semantic
  derivative. In the sets-of-words semantics, however, overlaps are possible, so
  we do not get a partition.}

\begin{proposition} 
\label{prop:antimirov-correct}
For any $E$, 
\begin{enumerate}
\item for any $a \in \Sigma$, $D_a \sem{E} = \bigcup \{\sem{E'} \mid E \to (a, E')\}$;
\item for any $u \in \Sigma^*$,
  $D_u \sem{E} =  \bigcup \{\sem{E'} \mid E \to^* (u, E')\}$.
\end{enumerate}
\end{proposition}

\begin{corollary} For any $E$,
\begin{enumerate}
%\item $\sem{E} = \{\eps \mid E \dn\} \cup \bigcup \{ \{a\} \cdot \sem{E'} \mid E \to (a, E') \}$;
\item $av \in \sem{E} = \{u \mid \exists E'.\, E \to (a, E') \con v \in \sem{E'}\}$.
\item $uv \in \sem{E} = \{u \mid \exists E'.\, E \to^* (u, E') \con v \in \sem{E'}\}$.
\item $u \in \sem{E} = \{u \mid \exists E'.\, E \to^* (u, E') \con E' \dn\}$.
\end{enumerate}
\end{corollary}

% \hendrik{big-step}
% \begin{lemma}
% For any $E$ and any $u \in \Sigma^*$,
% $E \Rightarrow u \mathbin{\mathrm{iff}} \exists E'.\, E \to^* (u, E') \con E' \dn
% $.
% \end{lemma}

The parts-of-derivatives of a regexp $E$ induce a nondeterministic
automaton. The state set is
$Q^E \eqdf \{E' \mid \exists u \in \Sigma^*.\, E \to^* (u, E')\}$.
The set of initial states is $I^E \eqdf \{E\}$. The set of final
states is $F^E \eqdf \{E' \in Q^E \mid E' \dn \}$. Finally, the
transition relation is defined by
$E' \to^E (a, E'') \eqdf E' \to (a, E'')$ for $E', E'' \in Q^E$.

The state set $Q^E$ is shown finite by proving it to be a subset of
another set that is straightforwardly seen to be finite.

\begin{definition} For any $E$, we define a set $E^{\to^*}$ of regexps
  recursively by
\[
\begin{array}{rcl}
E^{\to^*}       & \eqdf & \{E \} \cup E^{\to^+} \\ [2ex]

a^{\to^+}       & \eqdf & \{ 1 \} \\
0^{\to^+}       & \eqdf & \emptyset \\
(E + F)^{\to^+} & \eqdf & E^{\to^+} \cup F^{\to^+} \\
1^{\to^+}       & \eqdf & \emptyset \\
(EF)^{\to^+}    & \eqdf & E^{\to^+} \cdot \{F\} \cup F^{\to^+} \\
(E^*)^{\to^+}   & \eqdf & E^{\to^+} \cdot \{E^*\}
\end{array}
\]
\end{definition}

\begin{proposition} For any $E$,
\begin{enumerate}
\item $E^{\to^*}$ is finite, in fact, of cardinality linear in the
  size of $E$;
\item $Q^E \subseteq E^{\to^*}$.
\end{enumerate}
\end{proposition}

\begin{proof}
Both parts by induction on $E$.
\end{proof}

\begin{corollary}
  For any $E$, the Antimirov automaton is finite.
\end{corollary}

We note that the Antimirov automaton, constructed as above, while
canonical, is generally not trim: every state is reachable, but not
every state is generally coreachable (i.e., not every state needs to
have a path to some final state). A state $E'$ is not coreachable if
and only if $\sem{E'} = \emptyset$.  This is the case precisely when
$E'$ equals $0$ in the theory of idempotence of $+$ and the left and
right zero laws of $0$ wrt. $\cdot$.  The Antimirov automaton is
trimmed by removing the states that are not coreachable.

Now we can also show that a suitable quotient of the Brzozowski
automaton is finite.

For this we prove a syntactic version of
Proposition~\ref{prop:antimirov-correct} relating the Brzozowski
derivative and the Antimirov parts-of-derivatives. 

\begin{proposition} For any $E$, 
\begin{enumerate}
\item for any $a \in \Sigma$, $D_a E \doteq \sum \{E' \mid E \to (a, E')\}$;
\item for any $u \in \Sigma^*$,
  $D_u E \doteq \sum \{E' \mid E \to^* (u, E')\}$.
\end{enumerate}
(using the semilattice equations for $0, +$, that $0$ is left zero, and
distributivity of $\cdot$ over $+$ from the right).
\end{proposition}

\begin{corollary}
  For any $E$, the Brzozowski automaton, suitably quotiented, is
  finite.
\end{corollary}

\begin{proof}
  Just notice that the powerset of a finite set is finite too.
\end{proof}

This quotient does not give the minimal deterministic automaton (given
by semantic derivatives of $\sem{E}$). The minimal deterministic
automaton is obtained from the Brzozowski automaton by quotienting it
by the full Kleene algebra theory.

\section{Trace Closures of Regular Languages}

\subsection{Trace Closure of a Word Language}

An \emph{independence alphabet} is an alphabet $\Sigma$ together with
an irreflexive and symmetric relation
$I \subseteq \Sigma \times \Sigma$ called the \emph{independence}
relation. The complement $D$ of $I$, which is reflexive and symmetric,
is called \emph{dependence}.
We extend independence to words by saying that two words $u$ and $v$
are independent, $u I v$, if $a I b$ for all $a, b$ 
such that $a \in \Sigma(u)$ and $b \in \Sigma(v)$.

Let $\sim^I \subseteq \Sigma^* \times \Sigma^*$ be the least
congruence relation on the free monoid $\Sigma^*$ such that $aIb$
implies $ab \sim^I ba$ for all $a, b \in \Sigma$. If $u I v$, then
$uv \sim^I vu$.

A \emph{(Mazurkiewicz) trace} is an equivalence class of words wrt.\
$\sim^I$. The equivalence class of a word $w$ is denoted by $[w]^I$.

A word $a_1 \ldots a_n$ where $a_i \in \Sigma$ yields a directed
node-labelled acyclic graph as follows. Take the vertex set to be
$V \eqdf \{ 1, \ldots , n\}$ and label vertex $i$ with $a_i$. Take the
edge set to be $E \eqdf \{ (i, j) \mid i < j \con a_i D a_j \}$. This
graph $(V, E)$ for a word $w$ is called the \emph{dependence graph} of
$w$ and is denoted by $\langle w \rangle_D$. If $w \sim^I z$, then the
dependence graphs of $w$ and $z$ are isomorphic, i.e., traces can be
identified with dependence graphs up to isomorphism.

The set $\Sigma^*/{\sim^I}$ of all traces is the free partially
commutative monoid on $(\Sigma, I)$.  If $I = \emptyset$, then
$\Sigma^*/{\sim^I} \cong \Sigma^*$, the set of words, i.e., we
recover the free monoid. If $I = \{ (a,b) \mid a \neq b \}$, then
$\Sigma^*/{\sim^I} \cong \Mf(\Sigma)$, the set of finite multisets
over $\Sigma$, i.e., the free commutative monoid.

A \emph{trace language} is a subset of $\Sigma^*/{\sim^I}$.  Trace
languages are in bijection with word languages that are \emph{(trace)
  closed} in the sense that, if $z \in L$ and $w \sim^I z$, then also
$w \in L$. If $T$ is a trace language, then its flattening
$L \eqdf \bigcup T$ is a closed word language. On the other hand, the
trace language corresponding to a closed word language $L$ is
$T \eqdf \{t \in \Sigma^*/{\sim^I} \mid \exists z \in t.\, z \in L\} =
\{t \in \Sigma^*/{\sim^I} \mid \forall z \in t.\, z \in L\}$.

 % can be identified with word languages via flattening. If $T$
% is a trace language, then $L \eqdf \bigcup T$ is a word language that
% is \emph{(trace) closed} in the sense that, if
% $u \in L$ and $u \sim^I v$, then also $v \in L$. On the other hand,
% any word language $L$ determines a trace language
% $T \eqdf \{t \in \Sigma^*/{\sim^I} \mid \exists u \in t.\, u \in
% L\}$.
% This trace language satisfies $\bigcup T = L$ but also, e.g.,
% $T = \{t \in \Sigma^*/{\sim^I} \mid \forall u \in t.\, u \in L\}$.

Given a general (not necessarily closed) word language $L$, we define its
\emph{(trace) closure} $[L]^I$ as the least closed word language that
contains $L$. Clearly
$[L]^I = \{ w \in \Sigma^* \mid \exists z \in L.\, w \sim^I z \}$
and also
$[L]^I = \bigcup \{t \in \Sigma^*/{\sim^I} \mid \exists z \in t.\, z
\in L\}$.
For any $L$, we have $[[L]^I]^I = [L]^I$, so $[\_]^I$ is a closure
operator. Note also that $L$ is closed iff $[L]^I = L$.

As seen in Section~\ref{sec:derivatives-of-a-language}, the derivative
of a word language is the set of all suffixes for a prefix. We now
look at what the prefixes and suffixes of a word as a representative
of a trace should be. For a word $vuv'$ such that $v I u$, we can
consider $u$ to be its prefix, up to reordering, and $vv'$ to be the
suffix. This is because an equivalent word $uvv'$ strictly has $u$ as
a prefix and $vv'$ as the suffix. Similarly, we may also want to
consider $u'$ to be a prefix of $vuv'$ when $u' \sim^I u$ since
$u'vv' \sim^I uvv' \sim^I vuv'$. Note that if $a$ is such a prefix of
$z$, then, by irreflexivity of $I$, this $a$ is the first $a$ of
$z$. In general, when $u$ is a prefix of $z$, then the letter
occurrences in $u$ uniquely map to letter occurrences in $z$.  We
scale these ideas to allow $u$ to be scattered in $z$ as
$z = v_0u_1v_1 \ldots u_nv_n$ in either the sense that
$u = u_1 \ldots u_n$ or $u \sim^I u_1 \ldots u_n$. We also define
degree-bounded versions of scattering that become relevant in
Section~\ref{sec:uniform-scattering-rank}.

\begin{definition}
\label{def:scat}
 For all $u_1, \ldots , u_n \in \Sigma^+, v_0 \in \Sigma^*, v_1 \ldots , v_{n-1} \in \Sigma^+, v_n \in \Sigma^*, z
  \in \Sigma^*$, \\
 $u_1, \ldots, u_n \lhd z \rhd v_0, \ldots, v_n \eqdf  z = v_0u_1v_1 \ldots u_nv_n \con 
\forall i.\, \forall j < i.\,  v_j I u_i$.
\end{definition}

\begin{definition} For all $u, v, z \in \Sigma^*$,
\begin{enumerate}
\item $u \lhd z \rhd v \eqdf \exists n \in \Nat, u_1, \ldots , u_n, v_0, \ldots
  , v_n.\,
  u = u_1\ldots u_n \con v = v_0 \ldots v_n \con \\ \hspace*{5mm}  
   u_1, \ldots, u_n \lhd z \rhd v_0, \ldots, v_n$;  
\item $u \lsim z \rhd v \eqdf \exists u'.\, u \sim^I u' \con u' \lhd z \rhd v$; 
\item $u \lsim z \rsim v \eqdf \exists u', v'.\, u \sim^I u' \con u' \lhd z \rhd v' \con v' \sim^I v$.
\end{enumerate}
\end{definition}

In all three cases, we talk about $u$ being a prefix and $v$ being a
suffix of $z$, up to reordering, or $uv$ being \emph{scattered} in
$z$ with \emph{degree} $n$.

\begin{lemma} For any $u, v, z \in \Sigma^*$, 
\begin{enumerate}
\item $u \lhd z \rhd v \iff \exists! n \in \Nat, u_1, \ldots , u_n, v_0, \ldots
  , v_n.\,  u = u_1\ldots u_n \con v = v_0 \ldots v_n \con \\ \hspace*{5mm} 
  u_1, \ldots, u_n \lhd z \rhd v_0, \ldots, v_n$;
\item $u \lsim z \rhd v \iff \exists! u'.\, u \sim^I u' \con u' \lhd z \rhd v$; 
\item $u \lsim z \rsim v \iff \exists! u', v'.\, u \sim^I u' \con u' \lhd z \rhd v' \con v' \sim^I v$.
\end{enumerate}
\end{lemma}

\begin{definition} For all $u, v, z \in \Sigma^*$ and $N \in \Nat$,
\begin{enumerate}
\item $u \lhd_N z \rhd v \eqdf \exists n \leq N, u_1, \ldots , u_n, v_0, \ldots
  , v_n.\, u_1, \ldots, u_n \lhd z \rhd v_0, \ldots, v_n$; 
\item (and $u \lsim_N z \rhd v$ and $u \lsim_N z \rsim v$ are defined analogously).
\end{enumerate}
\end{definition}

\begin{example}
  Let $\Sigma \eqdf \{ a, b, c \}$ and $a I b$ and $a I c$. Take
  $z \eqdf aabcba$. We have $ab \lhd z \rhd acba$ since
  $a, b \lhd z \rhd \eps, a, cba$.  We can visualize this by
  underlining the subwords of $u \eqdf ab$ in
  $z = \eps \underline{a}a\underline{b}cba$.  This scattering is valid
  because $\eps I a$, $\eps I b$ and $a I b$: recall that
%$a, b \lhd z \rhd \eps, a, cba$ 
  Def.~\ref{def:scat} requires all underlined subwords $u_i$ to be
  independent with all non-underlined subwords $v_i$ to their left in
  $z$. Similarly we have $aa, a \lhd z \rhd \eps, bcb, \eps$ because
  $z = \eps \underline{aa}bcb\underline{a}\eps$, $\eps I aa$,
  $\eps I a$ and $bcb I a$. Note that neither
  $aa\underline{b}cb\underline{a}\eps$ nor $aabc\underline{ba}\eps$
  satisfies the conditions about independence and thus there is no $v$
  such that $ba \lhd z \rhd v$. We do have $ba \lsim z \rhd acba$
  though, since $ba \sim^I ab$ and $a, b \lhd z \rhd \eps, a, cba$.
\end{example}

\begin{proposition} \label{prop:scat-uv} For all $u, v, z \in \Sigma^*, 
  uv \sim^I z \iff u \lsim z \rsim v$.
\end{proposition}

%This is not exactly a generalization of Levi's lemma but it is
%similar.

\subsection{Trace-Closing Semantics of Regular Expressions}

%In this section, 
We now define a nonstandard word-language semantics of regexps that
directly interprets $E$ as the trace closure $[\sem{E}]^I$ of its
standard regular word-language denotation of $\sem{E}$.

We have $[\{a\}]^I = \{a\}$, $[\emptyset]^I = \emptyset$,
$[L \cup L']^I = [L]^I \cup [L']^I$ and $[\one]^I = \one$. But for
general $I$, we do not have $[L \cdot L']^I = [L]^I \cdot [L']^I$.
For example, for $\Sigma \eqdf \{a, b\}$ and $a I b$, we 
have $[\{a\}]^I = \{a\}$,  $[\{b\}]^I = \{b\}$ whereas
$[\{ab\}]^I = \{ ab, ba\} \neq \{ ab\} = [\{a\}]^I \cdot [\{b\}]^I$. 
Hence we need a different concatenation operation.

\begin{definition} ~
\begin{enumerate}
\item The $I$-\emph{reordering concatenation of words}
  $\cdot^I : \Sigma^* \times \Sigma^* \to \P \Sigma^*$ is defined by 
\[
\begin{array}{rcl}
\eps \cdot^I v   & \eqdf & \{ v \} \\
u \cdot^I \eps   & \eqdf & \{ u \} \\
a u \cdot^I b v  & \eqdf &
  \{a\} \cdot (u \cdot^I b v) \cup \{b \mid a u I b \} \cdot (a u \cdot^I v) \\[2ex]
\end{array}
\] 
\item The lifting of $I$-reordering concatenation to \emph{languages} is defined by 
\[
L \cdot^I L' \eqdf \bigcup \{ u \cdot^I v \mid u \in L \con v \in L' \}
\] 
\end{enumerate}
\end{definition}
Note that $\{ b \mid auIb \}$ acts as a test: it is either $\emptyset$
or $\{b\}$.

\begin{example}
  Let $\Sigma \eqdf \{ a , b \}$ and $a I b$. Then
  $a \cdot^I b = \{ ab, ba \}$, $aa \cdot^I b = \{ aab, aba, baa \}$,
  $a \cdot^I bb = \{ abb, bab, bba \}$ and
  $ab \cdot^I ba = \{ abba \}$. The last example shows that although
  $I$-reordering concatenation is defined quite similarly to shuffle,
  it is different. 
\end{example}

\begin{proposition} \label{prop:scat} For any $u, v, z \in \Sigma^*$,
  $z \in u \cdot^I v \iff u \lhd z \rhd v$.
\end{proposition}

\begin{proposition} \label{prop:reord-concat}
  For any languages $L$ and $L'$, $[L \cdot L']^I = [L]^I \cdot^I [L']^I$.
\end{proposition}

Evidently, if $I = \emptyset$, then reordering concatenation is just
ordinary concatenation: $u \cdot^\emptyset v = \{ uv \}$. 
% and $L \cdot^I L' = L \cdot L'$. 
For $I = \Sigma \times \Sigma$, which
is forbidden in independence alphabets, as $I$ is required to be
irreflexive, it is shuffle: $u \cdot^{\Sigma\times \Sigma} v = u \shuffle v$. 
For general $I$, it has properties similar
to concatenation. In particular, we have
\[
\begin{array}{rcl@{\quad}rcl}
\one \cdot^I L & = & L 
&  \emptyset \cdot^I L & = & \emptyset \\
L \cdot^I \one & = & L 
&  (L_1 \cup L_2) \cdot^I L & = & L_1 \cdot^I L \cup L_2 \cdot^I L  \\
(L \cdot^I L') \cdot^I L'' & = & L \cdot^I (L' \cdot^I L'') 
&  (L_1 \shuffle L_2) \cdot^I (L_1' \shuffle L_2') & \subseteq & 
   (L_1 \cdot^I L_1') \shuffle (L_2 \cdot^I L_2')%  \\
% \emptyset \cdot^I L & = & \emptyset \\
% (L_1 \cup L_2) \cdot^I L & = & L_1 \cdot^I L \cup L_2 \cdot^I L  \\
\end{array}
\]
but also other equations of the concurrent Kleene algebra theory
introduced in \cite{HoareMSW11}.

We are ready to introduce the closing semantics of regular
expressions.

\begin{definition} The \emph{trace-closing semantics}
  $\sem{\_}^I : \RE \to \P \Sigma^*$ of regular expressions is defined
  recursively by
\[
\small
\begin{array}{rcl@{\qquad}rcl}
\sem{a}^I & \eqdf & \{ a \} 
& \sem{1}^I & \eqdf & \one  \\
\sem{0}^I & \eqdf & \emptyset 
& \sem{E F}^I & \eqdf & \sem{E}^I \cdot^I \sem{F}^I \\
\sem{E + F}^I & \eqdf & \sem{E}^I \cup \sem{F}^I 
& \sem{E^*}^I & \eqdf & \mu X.\, \one \cup \sem{E}^I \cdot^I X%  \\
% \sem{1}^I & \eqdf & \one \\
% \sem{E F}^I & \eqdf & \sem{E}^I \cdot^I \sem{F}^I \\
% \sem{E^*}^I & \eqdf & \mu X.\, \one \cup \sem{E}^I \cdot^I X
\end{array}
\]
\end{definition}

Compared to the standard semantics of regular expressions, the
difference is in the handling of the $EF$ case (and consequently also
the $E^*$ case) due to the cross-commutation that happens in
concatenation of traces and must be accounted for by ${\cdot^I}$.

With $I = \emptyset$, we fall back to the standard interpretation of
regular expressions: $\sem{E}^\emptyset = \sem{E}$. For $I$ a general
independence relation, we obtain the desired property that the
semantics delivers the trace closure of the language of the regexp.

\begin{proposition}
  For any $E$, $\sem{E}^I$ is trace closed; moreover,
  $\sem{E}^I = [\sem{E}]^I$.
\end{proposition}

% We can also define a relational version of this semantics by adjusting
% the rules.
% \begin{definition} Big-step relational trace language semantics.
% \[
% \begin{array}{l}
% \infer{a \Rightarrow^I a}{}
% \quad
% \infer{E + F \Rightarrow^I u}{
%   E \Rightarrow^I u
% }
% \quad
% \infer{E + F \Rightarrow^I u}{
%   F \Rightarrow^I u
% } 
% \quad
% \infer{1 \Rightarrow^I \eps}{}
% \quad
% \infer{EF \Rightarrow^I w}{
%  E \Rightarrow^I u
%  & 
%  F \Rightarrow^I v
%  &
%  w \in u \cdot^I v
% }
% \\[2ex]
% \infer{E^* \Rightarrow^I \eps}{}
% \quad
% \infer{E^* \Rightarrow^I w}{
%  E \Rightarrow^I u
%  & 
%  E^* \Rightarrow^I v
%  &
%  w \in u \cdot^I v
% }
% \end{array}
% \]
% \end{definition}

% Again we have adequacy wrt. the big-step functional-style semantics.
% \begin{lemma} 
%   For every $E$, $\sem{E}^I = \{u \mid E \Rightarrow^I u \}$.
% \end{lemma}

\subsection{Properties of Trace Closures of Regular Languages}

Trace closures of regular languages are theoretically interesting due to
their intricate properties and have therefore been studied in a number
of works, e.g.,
\cite{BertoniMS82,Ochmanski85,AalbersbergW86,Sakarovitch87,Hashiguchi91,KlunderOS05}. For
a thorough survey, see Ochma\'nski's handbook
chapter~\cite{Ochmanski95}.

The most important property for us is that the trace closure of a regular language is
not necessarily regular. 

\begin{proposition} There exists a regular language $L$ such that $[L]^I$ is not regular.
\end{proposition}

\begin{proof}
Consider $\Sigma \eqdf \{a, b\}$, $a I b$.
Let $L \eqdf \sem{(ab)^*}$. The language
$[L]^I = \{ u \mid | u |_a = | u |_b \}$ is not regular.
\end{proof}

The class of trace closures of regular languages over an independence
alphabet behaves quite differently from the class of regular languages
over an alphabet. Here are some results demonstrating this. 

% For example, the class of trace closures of regular
% languages over $(\Sigma, I)$ is closed under complement iff $I$ is
% quasi-transitive (i.e., its reflexive closure is transitive)
% \cite{BertoniMS82,AalbersbergW86,Sakarovitch87}
% (cf.~\cite[Thm.~6.2.5]{Ochmanski95}); the question of whether the
% trace closure of the language of a regexp over $(\Sigma, I)$ is
% regular is decidable iff $I$ is quasi-transitive \cite{Sakarovitch87}
% (cf.~\cite[Thm.~6.2.7]{Ochmanski95}).

\begin{theorem}[Bertoni et al. \cite{BertoniMS86}, Aalbersberg and Welzl \cite{AalbersbergW86},
  Sakarovitch \cite{Sakarovitch87}]
  (cf.~\cite[Thm.~6.2.5]{Ochmanski95}) The class of trace closures of
  regular languages over $(\Sigma, I)$ is closed under complement iff
  $I$ is quasi-transitive (i.e., its reflexive closure is transitive).
\end{theorem}

\begin{theorem}[Bertoni et al. \cite{BertoniMS82}, Aalbersberg and
  Welzl \cite{AalbersbergW86} (``if'' part); Aalbersberg and Hoogeboom
  \cite{AalbersbergH89}]
  (cf.~\cite[Thm.~6.2.5]{Ochmanski95}) The problem of whether the
  trace closures of two regular languages over $(\Sigma, I)$ are equal
  is decidable iff $I$ is quasi-transitive.
\end{theorem}

\begin{theorem}[Sakarovitch \cite{Sakarovitch92}] 
  (cf.~\cite[Thm.~6.2.7]{Ochmanski95}) The problem of whether the
  trace closure of the language of a regexp over $(\Sigma, I)$ is
  regular is decidable iff $I$ is quasi-transitive.
\end{theorem}

A closed language is regular iff the corresponding trace language is
accepted by a finite asynchronous (a.k.a.\ Zielonka)
automaton~\cite{Zielonka87,Zielonka95}. In Section~\ref{subsec:starconn}, we
will see further characterizations of regular closed languages based
on star-connected expressions.

\subsection{Rational and Recognizable Languages of Monoids}

Trace languages are a special case of languages of monoids. 

A subset $T$ of a monoid $M$ is called an $M$-language. 

An $M$-language $T$ is called \emph{rational} if $ T = \sem{E}^M$ for
some regular expression $E$ over $M$. Here
$\sem{\_}^M : \RE(M) \to \P M$ interprets any element $m$ of $M$ as
$\{m\}$, the $0, +$ constructors of regular expressions by $\emptyset$
and $\cup$, the $1, {\cdot}$ constructors as mandated by the monoid
structure, and $(\_)^*$ as the appropriate least fixpoint.

An $M$-language $T$ is called \emph{recognizable} if there is a
deterministic finite $M$-automaton accepting $T$. An deterministic
$M$-automaton is given by a state set $Q$, an initial state
$q_0 \in Q$, a set of final states $F \subseteq Q$, and a right action
$\delta$ of $M$ on $Q$. An element $m \in M$ is accepted by the
automaton if $\delta_m q_0 \in F$.

Kleene's celebrated theorem says that, for languages of free monoids
on finite sets (i.e., word languages over finite alphabets),
rationality and recognizability are equivalent conditions (and we can
thus just speak about regularity). For a general monoid, however, the
two notions are different.

\begin{theorem}[Kleene~\cite{Kleene56}] Let $M$ be the free monoid
  $\Sigma^*$ on a finite set $\Sigma$. An $M$-language $T$ is rational
  iff $T$ is recognizable.
\end{theorem}

\begin{theorem}[McKnight~\cite{McKnight64}] Let $M$ be finitely
  generated. If an $M$-language $T$ is recognizable, then $T$ is
  rational.
\end{theorem}

Given a monoid $M$ and a congruence $\equiv$ on $M$, the set
$M/{\equiv}$ is a monoid too. We view $M/{\equiv}$-languages as sets of
equivalence classes wrt.\, $\equiv$.

\begin{proposition}
\label{prop:rat-rec}
  Given a monoid $M$ and a congruence $\equiv$ on it.
\begin{enumerate}
\item The $M/{\equiv}$-language $\sem{E}^{M/{\equiv}}$ of a regular
  expression $E$ is expressible via its $M$-language $\sem{E}^M$ by
  $\sem{E}^{M/{\equiv}} = \{t \in M/{\equiv} \mid \exists u \in t.\, u \in \sem{E}^M \}$.
 %
 % or, equivalently,
 % $\bigcup \sem{E}^{M/\equiv} = \{u' \mid \exists u.\, u' \equiv u
 % \con u \in \sem{E}^M\}$.

\item A $M/{\equiv}$-language $T$ is recognizable iff its flattening
  $\bigcup T$ into an $M$-language is recognizable.
\end{enumerate}
\end{proposition}

% \subsection{Rational and Recognizable Trace Languages in Terms of Word Languages}

% Given an independence alphabet $(\Sigma, I)$. For the monoid
% $\Sigma^*$ of words, which is the free monoid on $\Sigma$, the classes
% of rational and recognizable languages coincide (and we simply speak
% of regular languages).

For the monoid $\Sigma^*/{\sim^I}$ of traces, which is the free
partially commutative monoid, the classes of rational and recognizable
languages are different, the class of rational languages is a proper
subclass of that of recognizable languages. In view of
Proposition~\ref{prop:rat-rec}, a trace language $T$ is rational iff
$T = \{t \in \Sigma^*/{\sim^I} \mid \exists u \in t.\, u \in L\}$ or, equivalently,
$\bigcup T = [L]^I$ for some regular word language $L$ (in the
alternative terminology of Aalbersberg and Welzl \cite{AalbersbergW86},
such a trace language $T$ is called \emph{existentially regular}), and
recognizable iff $\bigcup T = L$ for some regular word language $L$
(such a trace language is called \emph{consistently regular}).

The question of when a rational trace language is recognizable is
nontrivial. We have just seen that, reformulated in terms of word
languages, it becomes: given a regular language $L$, when is its trace
closure $[L]^I$ regular?

\section{Reordering Derivatives}

We are now ready to generalize the Brzozowski and Antimirov
constructions for trace closures of regular languages. To this end, we
switch to what we call reordering derivatives.

\subsection{Reordering Derivative of a Language}

Let $(\Sigma, I)$ be a fixed independence alphabet. We generalize the
concepts of (semantic) nullability and derivative of a
language to concepts of reorderable part and reordering
derivative. 

\begin{definition} \label{def:reord-deriv} We define the
  \emph{$I$-reorderable part} of a language $L$ wrt.\ a word $u$ by
  $\R^I_u L \eqdf \{ v \in L \mid v I u \} $ and the
  \emph{$I$-reordering derivative} along $u$ by
  $D^I_u L \eqdf \{v \mid \exists z \in L.\, u \lsim z \rhd v\}$.
\end{definition}
By Prop.~\ref{prop:scat}, we can equivalently say that
$D^I_u L = \{v \mid \exists z \in L.\, z \in [u]^I \cdot^I v \}$. For
a single-letter word $a$, we get
$D^I_a L = \{ v_lv_r \mid v_lav_r \in L \con v_l I a \} = \{v \mid
\exists z \in L.\, z \in a \cdot^I v \}$.
That is, we require some reordering of $u$ (resp.\ $a$) to be a
prefix, up to reordering, of some word $z$ in $L$ with $v$ as the
corresponding strict suffix.  (In other words, for the sake of
precision and emphasis, we allow reordering of letters within $u$ and
across $u$ and $v$, but not within $v$.)

\begin{example}
  Let $\Sigma \eqdf \{ a, b, c \}$ and $a I b$. Take
  $L \eqdf \{ \eps, a, b, ca, aa, bbb, babca, abbaba \}$. We have
  $\R^I_{a}L=\R^I_{aa}L = \{ \eps, b, bbb \}$,
  $D^I_{a}L = \{ \eps, a, bbca, bbaba \}$ and $D^I_{aa}L = \{ \eps, bbba \}$.
\end{example}

In the special case $I = \emptyset$, we have
$\R^\emptyset_\eps L = L$, $\R^\emptyset_u L = \{\eps \mid L \dn \}$
for any $u \neq \eps$, and $D^\emptyset_u L = D_u L$.
In the general case, the reorderable part and reordering
derivative enjoy the following properties.

\begin{lemma} For every $L$, $L'$, for any $u \in \Sigma^*$, if
  $L \subseteq L'$, then $R^I_u L \subseteq R^I_u L'$ and
  $D^I_u L \subseteq D^I_u L'$.
\end{lemma}

\begin{lemma} For every $L$,
  \begin{enumerate}
%  \item for every $u \in \Sigma^*$, $R^I_u L \subseteq L$ and 
%      $R^I_u (R^I_u L) = R^I_u L$;
  \item $R^I_\eps L = L$; for every $u, v \in \Sigma^*$,
    $R^I_v (R^I_u L) = R^I_{uv} L$;
%  \item for every $u, v \in \Sigma^*$,
%    $R^I_v (R^I_u L) = R^I_u (R^I_v L)$.
  \item for every $u, u' \in \Sigma^*$, $R^I_{\Sigma(u)} L = R^I_{\Sigma (u')} L$.
  \end{enumerate}
\end{lemma}

%By the previous lemma we have that $\Sigma(u) = \Sigma(v)$
%  implies $R^I_u L = R^I_v L$. 
We extend $\R^I$ to subsets of $\Sigma$: by $\R^I_X L$, we mean
$\R^I_u L$ where $u$ is any enumeration of $X$.

% \begin{lemma}~
% \item For every $a$, $\N L = \bigcap_{a \in \Sigma}
%   \R^I_a L$.
% \item For every $a$, $\R^{\emptyset}_a L = \N L$. 
% \end{lemma}

% \begin{proposition} For every $u, v \in \Sigma^*$,
%   $L \subseteq \Sigma^*$, $v \in D^I_u L$ iff
%   $\exists z \in L . \, u \lsim z \rhd v$.
% \end{proposition}

%\hendrik{From the above proposition it follows that the reordering
%  derivative is a trace monoid action, i.e., if $u \sim^I u'$, then
%  $D^I_u L = D^I_{u'} L$.}

% \begin{lemma} For every $L$, % u, v \in \Sigma^*$,
% \begin{enumerate}
% \item for any $u, v \in \Sigma^*$,  $v \in D^I_u L$ iff $\exists z \in L . \,  u \lsim z \rhd v$;
% \item for any $u, u' \in \Sigma^*$, if $u \sim^I u'$, then $D^I_u L = D^I_{u'} L$.
% \end{enumerate}
% \end{lemma}

\begin{lemma} For every $L$,
\begin{enumerate}
\item $D^I_\eps
  L = L$; for any $u, v \in \Sigma^*$, $D^I_v (D^I_u L) = D^I_{uv} L$;
% \item for any $u, v\in \Sigma$ such that $u I v$, we have
%   $D^I_v (D^I_u L) = D^I_u (D^I_v L)$;
\item for any $u,
  u' \in \Sigma^*$ such that $u \sim^I u'$, we have $D^I_u L =
  D^I_{u'} L$.
\end{enumerate}
\end{lemma}

\begin{proposition} For every $L$, 
\begin{enumerate}
\item for any  $u \in \Sigma^*$, $D_u ([L]^I) = [D^I_u L]^I$; \\
  if $L$ is closed (i.e., $[L]^I = L$), then, for any
  $u \in \Sigma^*$, $D^I_u L$ is closed and $D_u L = D^I_u L$;
\item for any $u, v \in \Sigma^*$, $uv \in [L]^I$ iff $v \in [D^I_u L]^I$;
\item for any  $u \in \Sigma^*$, $u \in [L]^I$ iff $(D^I_u L) \dn$; 
\item $[L]^I = \{\eps \mid L \dn\} \cup \bigcup_{a \in \Sigma} \{a\} \cdot [D^I_a L]^I$.
\end{enumerate}
\end{proposition}

% \begin{corollary} For every $L$, $I$ such that $L = [L]^I$,
% \begin{enumerate}
% \item for any $a \in \Sigma$,  $D_a L = D^I_a L$;
% \item $L = {\eps \mid \N L} \cup \bigcup_{a \in \Sigma} \{a\} \cdot (D^I_a L)$;
% \item $L = \{u \in \Sigma^* \mid N(D^I_u L)\}$.
% \end{enumerate}
% \end{corollary}

\begin{example}\label{ex:sem-deriv-infinite} %Expression with infinitely many derivatives. \\
  Let $\Sigma \eqdf \{ a , b \}$ and $a I b$. Take $L$ to be the regular
  language $\sem{(ab)^*}$. We have already noted that the language
  $[L]^I = \{ u \mid | u |_a = | u |_b \}$ is not regular.  For any
  $n \in \Nat$, $D^I_{b^n} L = \{a^n\} \cdot L = \sem{a^n(ab)^*}$ whereas
  $D_{b^n} ([L]^I) = \{a^n\} \cdot^I [L]^I = \{ u \mid | u |_a = | u |_b +
  n\}$.
  We can see that $[L]^I$ has infinitely many derivatives, none of which are regular, and
  $L$ has infinitely many reordering derivatives, all regular.
%  For $n \neq m$, we have $\{a^n\} \cdot L \neq \{a^m\} \cdot L$.
\end{example}

\subsection{Brzozowski Reordering Derivative}

The reorderable parts and reordering derivatives of regular languages
turn out to be regular. We now show that they can be computed
syntactically, generalizing the classical syntactic nullability and
Brzozowski derivative operations \cite{Brzozowski64}.
%that we review in Appendix~\ref{app:b-a}.

\begin{definition} The \emph{$I$-reorderable part} and the
  \emph{Brzozowski $I$-reordering derivative} of a regexp are given by
  functions $\R^I, D^I : \RE \times \Sigma \to \RE$ and
  $\R^I, D^I : \RE \times \Sigma^* \to \RE$ defined recursively by
\[
\small
\begin{array}{rcl@{\qquad}rcl}
R^I_a b & \eqdf & \mathsf{if~} a I b \mathsf{~then~} b \mathsf{~else~} 0 & 
D^I_a b & \eqdf & \mathsf{if~} a = b \mathsf{~then~} 1 \mathsf{~else~} 0 \\
R^I_a 0 & \eqdf & 0 &
D^I_a 0 & \eqdf & 0 \\
R^I_a (E + F) & \eqdf & R^I_a E + R^I_a F &
D^I_a (E + F) & \eqdf & D^I_a E + D^I_a F \\
R^I_a 1 & \eqdf & 1 &
D^I_a 1 & \eqdf & 0 \\
R^I_a (E F) & \eqdf & (R^I_a E) (R^I_a F) &
D^I_a (E F) & \eqdf & (D^I_a E) F + (R^I_a E) (D^I_a F) \\  
R^I_a (E^*) & \eqdf & (R^I_a E)^* &
D^I_a (E^*) & \eqdf & (R^I_a E)^* (D^I_a E) E^* \\[2ex]
R^I_\eps E & \eqdf & E &
D^I_\eps E & \eqdf & E \\
R^I_{ua} E & \eqdf & R^I_a (R^I_u E) &
D^I_{ua} E & \eqdf & D^I_a (D^I_u E) 
\end{array}
\]
\end{definition}

The regexp $R_u E$ is nothing but $E$ with all occurrences of letters
dependent with $u$ replaced with $0$. The definition of $D$ is more
interesting. Compared to the classical Brzozowski derivative, the
nullability condition $E \dn$ in the $EF$ case has been replaced with
concatenation with the reorderable part $\R^I_a E$, and the $E^*$ case
has also been adjusted.

The functions $R$ and $D$ on regexps compute their
semantic counterparts on the corresponding regular languages.

\begin{proposition}\label{prop:brzozowski-sem-deriv} For any $E$,
\begin{enumerate}
\item for any $a \in \Sigma$, $R^I_a \sem{E} = \sem{R^I_a E}$ and $D^I_a \sem{E} = \sem{D^I_a E}$;
\item for any $u \in \Sigma^*$, $R^I_u \sem{E} = \sem{R^I_u E}$ and $D^I_u \sem{E} = \sem{D^I_u E}$.
\end{enumerate}
\end{proposition}

\begin{proposition}\label{prop:brz-autom} For any $E$,
\begin{enumerate}
%\item $\sem{E}^I = \{\eps \mid E \dn\} \cup \bigcup_{a \in \Sigma} \{ a \} \cdot \sem{D^I_a E}^I$;
\item for any $a \in \Sigma$, $v \in \Sigma^*$, $av \in \sem{E}^I \iff v \in \sem{D^I_a E}^I$;
\item for any $u, v \in \Sigma^*$, $uv \in \sem{E}^I \iff v \in \sem{D^I_u E}^I$;
\item for any $u \in \Sigma^*$, $u \in \sem{E}^I \iff (D^I_u E) \dn$.
%\sem{E}^I = \{u \in \Sigma^* | N(D^I_u E)\}$. 
\end{enumerate}
\end{proposition}

\begin{example}
\label{ex:brz}
Let $\Sigma \eqdf \{ a, b \}$, $a I b$ and $E \eqdf aa + ab + b$.
\[
\begin{array}{rcl}
\small
D^I_b E & = & D^I_b aa + D^I_b ab + D^I_b b \\
        & = & ((D^I_b a)a + (R^I_b a)(D^I_b a)) + ((D^I_b a)b + (R^I_b a)(D^I_b b)) + D^I_b b \\
        & = & (0a + a0) + (0b + a1) + 1 
        \doteq  a + 1 \\[2ex]

D^I_b (E^*) & = & (R^I_b E)^*(D^I_b E)E^* \\
          % & = & (R^I_b aa + R^I_b ab + R^I_b b)^*(a + 1)E^* \\
           & = & (aa + a0 + 0)^*((0a + a0) + (0b + a1) + 1)E^* 
           \doteq  (aa)^*(a + 1)E^* \\[2ex]

D^I_{bb} (E^*) & \doteq & D^I_b ((aa)^*(a + 1)E^*) 
              \doteq (aa)^*(a + 1)(aa)^*(a + 1)E^*
\end{array}
\]
\end{example}

As with the classical Brzozowski derivative, we can use the reordering
Brzozowski derivative to construct deterministic automata. For a
regexp $E$, take $Q^E \eqdf \{D^I_u E \mid u \in \Sigma^*\}$,
$q^E_0 \eqdf E$, $F^E \eqdf \{E' \in Q^E \mid E' \dn\}$,
$\delta^E_a E'\eqdf D^I_a E'$ for $E' \in Q^E$. By
Prop.~\ref{prop:brz-autom}, this automaton accepts the closure
$\sem{E}^I$. But even quotiented by the full Kleene algebra theory,
the quotient of $Q^E$ is not necessarily finite, i.e., we may be able
to construct infinitely many different languages by taking reordering
derivatives. For the regexp from Example~\ref{ex:sem-deriv-infinite},
we have $D^I_{b^n} ((ab)^*) \doteq a^n(ab)^*$, so it has infinitely many
Brzozowski reordering derivatives even up to the Kleene algebra
theory. This is only to be expected, as the closure $\sem{(ab)^*}^I$
is not regular and cannot possibly have an accepting finite automaton.

% The language of Example~\ref{ex:sem-deriv-infinite} has infinitely many semantic reordering derivatives. Therefore, by Prop.~\ref{prop:brzozowski-sem-deriv}, the regexp defining it has infinitely many Brzozowski reordering derivatives up to the Kleene algebra theory. 
% This is only to be expected, as the closure is not regular and cannot possibly have an accepting finite automaton.

% By Prop.~\ref{prop:brzozowski-sem-deriv} and by completeness of the Kleene algebra axioms, it must also have there must also be infinitely many syntactic reordering derivatives.

\subsection{Antimirov Reordering Derivative}\label{sec:antimirov-expr}

Like the classical Brzozowski derivative that was optimized by
Antimirov \cite{Antimirov96},
%(cf.\ Appendix~\ref{app:b-a}), 
the Brzozowski reordering derivative construction can be optimized by
switching from functions on regexps to multivalued functions or
relations.

\begin{definition} The \emph{Antimirov $I$-reordering parts-of-derivatives}
  of a regexp along a letter and a word are relations
  ${\to^I} \subseteq \RE \times \Sigma \times \RE$ and
  ${\to^{I*}} \subseteq \RE \times \Sigma^* \times \RE$ defined
  inductively by
\[
\small
\begin{array}{l}
\infer{a \to^I (a, 1)}{
}
\quad
\infer{E + F \to^I (a, E')}{
  E \to^I (a, E')
}
\quad
\infer{E + F \to^I (a, F')}{
  F \to^I (a, F')
}
\\[2ex]
\infer{E F \to^I (a, E' F)}{
  E \to^I (a, E')
}
\quad
\infer{E F \to^I (a, (R^I_a E) F')}{
  F \to^I (a, F')
}
\quad
\infer{E^* \to^I (a, (R^I_a E)^* E' E^*)}{
  E \to^I (a, E')
}
\\[2ex]
\infer{E \to^{I*} (\eps, E)}{
}
\quad
\infer{E \to^{I*} (ua, E'')}{
  E \to^{I*} (u, E')
  &
  E' \to^I (a, E'')
}
\end{array}
\]
\end{definition}
Here $\R^I$ is defined as before. Similarly to the Brzozowski reordering
derivative from the previous subsection, the condition $E\dn$
in the second $EF$ rule has has been replaced with concatenation with
$\R^I_a E$, and the $E^*$ rule has been adjusted.

Collectively, the Antimirov reordering parts-of-derivatives of a
regexp $E$ compute the semantic reordering derivative of the
language $\sem{E}$.

\begin{proposition}\label{prop:antimirov-sem-deriv} For any $E$,
\begin{enumerate}
\item for any $a \in \Sigma$,
  $D^I_a \sem{E} = \bigcup \{ \sem{E'} \mid E \to^{I} (a, E') \}$; %\\
  % (i.e., for
  % any $a \in \Sigma$, $v \in \Sigma^*$,
  % $\exists v_l, v_r.\,  v_l a v_r \in \sem{E} \con v = v_lv_r \con v_l I a 
  % \iff \exists E'.\, E
  % \to^{I} (a, E') \con v \in \sem{E'}$);

\item for any $u \in \Sigma^*$,
  $D^I_u \sem{E} = \bigcup \{ \sem{E'} \mid E \to^{I*} (u, E') \}$. %\\
  % (i.e., for
  % any $u, v \in \Sigma^*$,
  % $\exists z \in \sem{E}.\, u \lsim z \rhd v \iff \exists E'.\, E
  % \to^{I*} (u, E') \con v \in \sem{E'}$).
\end{enumerate}
\end{proposition}

\begin{proposition} 
\label{prop:antim-autom}
For any $E$, 
\begin{enumerate} 
\item for any $a \in \Sigma$, $v \in \Sigma^*$, 
$av \in \sem{E}^I \iff \exists E'.\, E \to^I (a, E') \con v \in \sem{E'}^I$;
\item for any $u, v \in \Sigma^*$, 
$uv \in \sem{E}^I \iff \exists E'.\, E \to^{I*} (u, E') \con v \in
\sem{E'}^I$;
\item for any $u \in \Sigma^*$, 
$u \in \sem{E}^I \iff \exists E'.\, E \to^{I*} (u, E') \con E' \dn$.
\end{enumerate}
\end{proposition}

\begin{example}
  %Let $\Sigma \eqdf \{ a, b \}$, $a I b$ and $E \eqdf aa + ab + b$.
  Let us revisit Example~\ref{ex:brz}.
  The Antimirov reordering parts-of-derivatives of $E$ along $b$ are
  $a1$ and $1$:
\[
\small
\begin{array}{l}
%  \infer{E^* \to^I (b, (aa)^*(a1)E^*)}{
  \infer{aa + ab + b \to^I (b, a1)}{
  \infer{ab + b \to^I (b, a1)}{
  \infer{ab \to^I (b, a1)}{
  \infer{b \to^I (b, 1)}{}
  }
  }
  }
%  }
  \qquad
%  \infer{E^* \to^I (b, (aa)^*1E^*)}{
  \infer{aa + ab + b \to^I (b, 1)}{
  \infer{ab + b \to^I (b, 1)}{
  \infer{b \to^I (b, 1)}{}
  }
  }
%  }
\end{array}
\]
The Antimirov reordering parts-of-derivatives of $E^*$ along $b$ are
therefore $E_b^*(a1)E^*$ and $E_b^*1E^*$ where
$E_b \eqdf R^I_b E = aa + a0 + 0$. Recall that, for the Brzozowski
reordering derivative, we computed
$D^I_b E = (0a + a0) + (0b + a1) + 1$ and
$D^I_b E^* = E_b^* ((0a + a0) + (0b + a1) + 1)E^*$.
\end{example}

Like the classical Antimirov construction, the Antimirov reordering
parts-of-derivatives of a regexp $E$ give a nondeterministic automaton
by
$Q^E \eqdf \{E' \mid \exists u \in \Sigma^*.\, E \to^{I*}(u, E')\}$,
$I^E \eqdf \{E\}$, $F^E \eqdf \{E' \in Q^E \mid E' \dn\}$,
$E' \to^E (a, E'') \eqdf E' \to^I (a, E'')$ for $E', E'' \in Q^E$.
This automaton accepts $\sem{E}^I$ by Prop.~\ref{prop:antim-autom},
but is generally infinite, also if quotiented by the full Kleene
algebra theory.
Revisiting Example \ref{ex:sem-deriv-infinite} again, $(ab)^*$ must
have infinitely many Antimirov reordering parts-of-derivatives modulo
the Kleene algebra theory since $\sem{(ab)^*}^I$ is not regular and
cannot have a finite accepting nondeterministic
automaton. Specifically, it has
$(a0)^*((a1) \ldots ((a0)^*((a1)(ab)^*)) \ldots ) \doteq a^n(ab)^*$
%$((a0)^*(a1))^n(ab)^* \doteq a^n(ab)^*$ 
as its single reordering part-of-derivative along $b^n$.

However, if quotienting the Antimirov automaton for $E$ by some sound
theory (a theory weaker than the Kleene algebra theory) makes it
finite, then the Brzozowski automaton can also be quotiented to become
finite.
%, and we
%do not need the powerautomaton construction for this.

\begin{proposition}\label{prop:brzozowski-antimirov-equivalence} For any $E$,
\begin{enumerate}
\item for any $a \in \Sigma$,
  $D^I_a E \doteq \sum \{ E' \mid E \to^{I} (a, E') \}$; %\\
  % (i.e., for
  % any $a \in \Sigma$, $v \in \Sigma^*$,
  % $\exists v_l, v_r.\,  v_l a v_r \in \sem{E} \con v = v_lv_r \con v_l I a 
  % \iff \exists E'.\, E
  % \to^{I} (a, E') \con v \in \sem{E'}$);

\item for any $u \in \Sigma^*$,
  $D^I_u E \doteq \sum \{ E' \mid E \to^{I*} (u, E') \}$ 
\end{enumerate}
(using the semilattice equations for $0, +$, that $0$ is zero, and
distributivity of $\cdot$ over $+$).
\end{proposition}

\begin{corollary}
  If some quotient of the Antimirov automaton for $E$ (accepting
  $\sem{E}^I$) is finite, then also some quotient of the Brzozowski
  automaton is finite.
\end{corollary}

\subsection{Star-Connected Expressions}
\label{subsec:starconn}

Star-connected expressions are important as they characterize regular
closed languages. A corollary of that is a further characterization of
such languages in terms of a ``concurrent'' semantics of regexps that
interprets Kleene star nonstandardly as ``concurrent star''.
%(see Theorem~\ref{thm:star} in Appendix \ref{app:tc}). 

%Informally, an expression $E$ is said to be star-connected (wrt. $D$) if star is only used over connected languages.

\begin{definition}
  A word $w \in \Sigma^*$ is connected if its dependence graph
  $\langle w \rangle_D$ is connected. A language
  $L \subseteq \Sigma^*$ is connected if every word $w \in L$ is
  connected.
\end{definition}

\begin{definition} ~
\begin{enumerate}
\item \emph{Star-connected expressions} are a subset of the set of all regexps defined inductively by:
  $0$, $1$ and $a \in \Sigma$ are star-connected. If $E$ and $F$ are
  star-connected, then so are $E + F$ and $EF$. If $E$ is
  star-connected and $\sem{E}$ is connected, then $E^*$ is
  star-connected.
\item A language $L$ is said to be \emph{star-connected} if $L = \sem{E}$ for some star-connected regexp.
\end{enumerate}
\end{definition}

% \begin{proposition}
%   Let $L \subseteq \Sigma^*$ such that $L = [L]^I$. \\
%   $L$ is regular iff $L = \sem{E}^I$ for some star-connected
%   expression $E$.
% \end{proposition}

Ochma\'nski \cite{Ochmanski85} proved that a closed language is
regular iff it is the closure of a star-connected language.  This
means that, for any regexp $E$, the language $\sem{E}^I$ is regular
iff there exists a star-connected expression $E'$ such that
$\sem{E}^I = \sem{E'}^I$. It is important to realize that generally
$E \neq E'$ and also $\sem{E} \neq \sem{E'}$. Ochma\'nski's proof was
as follows.

For a linear order $\leq$ on $\Sigma$, a word $z$ is a
\emph{lexicographic normal form} if
$\forall w \in [z]^I.\, z \leq_{\mathrm{lex}} w$ where
$\leq_{\mathrm{lex}}$ is the lexicographic order on $\Sigma^*$ induced
by $\leq$. We write $\Lex^I$ for the set of all lexicographic normal
forms.

% \begin{lemma} ~
% \begin{enumerate}
% % \item For any $t \in \Sigma^*/\sim^I$, the set $t \cap \Lex^I$
% %   consists of exactly one element.
% \item For any $T \subseteq \Sigma^*/\sim^I$, if $T$ is recognizable, then
% $(\bigcup T) \cap \Lex^I \subseteq \Sigma^*$ is regular. 
% \item For any regular expression $E$, if $\sem{E} \subseteq \Lex^I$, then
%   $E$ is star-connected.
% \end{enumerate}
% \end{lemma}

% \begin{theorem}
%   For any $T \subseteq \Sigma^*/\sim^I$, the following are equivalent:
% \begin{itemize}
% \item $T$ is recognizable;
% \item $(\bigcup T) \cap \Lex^I$ is regular;
% \item $T = \sem{E}^{\Sigma^*/\sim^I}$ (i.e., $\bigcup T = \sem{E}^I$)
%   for some star-connected regular expression $E$.
% \end{itemize}
% \end{theorem}

\begin{lemma} (cf.\ \cite[Props.~6.3.4,
  6.3.10]{Ochmanski95}) \label{lem:ochmanski}
\begin{enumerate}
% \item For any $t \in \Sigma^*/\sim^I$, the set $t \cap \Lex^I$
%   consists of exactly one element.
\item $\Lex^I$ is regular. 
\item For any regular expression $E$, if $\sem{E} \subseteq \Lex^I$, then
  $E$ is star-connected.
\end{enumerate}
\end{lemma}

\begin{theorem}[Ochma\'nski \cite{Ochmanski85}]
  (cf.\ \cite[Thm.~6.3.13]{Ochmanski95}) \label{thm:star} For any closed
  language $L$, the following are equivalent:
\begin{enumerate}
\item $L$ is regular;
\item $L \cap \Lex^I$ is regular;
\item $L$ is star-connected.
\end{enumerate}
\end{theorem}

\begin{proof}
  (1) $\Rightarrow$ (2) is a consequence of Lemma~\ref{lem:ochmanski}(1)
  as the intersection of regular languages is regular. (2) $\imp$ (3)
  follows from Lemma~\ref{lem:ochmanski}(2). For (3) $\imp$ (1),
  Ochma\'nski employed Hachiguchi's notion of rank of a language and
  Hachiguchi's lemma, which we will study in Def.~\ref{def:rank} and
  Prop.~\ref{prop:hashiguchi} below, and proved that, if $L$ is closed
  and connected, then $L^*$ has rank.
\end{proof}

The nonstandard \emph{concurrent-star trace-language semantics} of
regular expressions
$\sem{\_}^\mathrm{con} : \RE(\Sigma) \to \mathcal{P} \Sigma^*$ is like
$\sem{\_}$ except that the star constructor is interpreted
nonstandardly as the \emph{concurrent star} operation. Informally, the
concurrent star of a language iterates not the given language but the
language of connected components of its words.

The concurrent star of a connected language coincides with its Kleene
star. The idea of this nonstandard semantics is to make
non-star-connected regular expressions harmless, so as to obtain the
following replacement for Kleene's theorem.

\begin{theorem}[Ochma\'nski \cite{Ochmanski85}]
  (cf.~\cite[Thm.~6.3.16]{Ochmanski95}) A closed langugage $L$ is regular iff
  $L = [\sem{E}^\mathrm{con}]^I$ for some regexp $E$.
\end{theorem}

\subsection{Automaton Finiteness for Star-Connected Expressions}

We now show that the set of Antimirov reordering parts-of-derivatives
of a star-connected expression is finite modulo suitable equations.

%Then also the set of Brzozowski reordering derivatives modulo
%$\doteq$ is finite by the powerset construction.

% \begin{lemma}
%   If $\sem{E}$ is connected, then, for every $a$ and $E'$
%   such that $E \to^I (a, E')$, either $R^I_a E' \doteq 0$ or
%   $R^I_a E' \doteq 1$ (using the equations involving $0$ and $1$ only
%   and that $0$ is zero).
% \end{lemma}
% % \begin{proof}

% % \end{proof}

% \begin{lemma} \label{lem:star-conn-0-1} If $\sem{E}$ is connected,
%   then, for every $n$, $a$ and $E'$ such that $E \to^{I*} (a^{n + 1}, E')$,
%   either $R^I_a E' \doteq 0$ or $R^I_a E' \doteq 1$ (using the
%   equations involving $0$ and $1$ only and that $0$ is zero).
% \end{lemma}

% \begin{proof}
%   \hendrik{Korda teha!}
%   Since $\Sigma(aub) = \Sigma(b)$ then it must be that $a = b$ and
%   $u = a^n$ for some $n$. Thus we can write
%   $E^* \to^{I*} (aa^na, (R^I_a E)^* (R^I_a E_{aa^n}) (R^I_a E)^* E_a
%   E^*)$.

%   If $v \in \sem{R^I_a E_{aa^n}}$, then $v \in \sem{E_{aa^n}}$ and
%   thus there is $z \in \sem{E}$ such that $aa^nv \sim^I z$. Since $z$
%   is a connected word, then so is $aa^nv$. Now, if $v$ is nonempty,
%   then there is $c \in v$ such that $a D c$. This cannot be since
%   $v \in \sem{R^I_a E_{aa^n}}$ and thus by definition $v I a$. This
%   means that if $w \in \sem{R^I_a E_{aa^n}}$, then $w = \eps$. Either
%   $\sem{R^I_a E_{aa^n}} = \emptyset$ or $\sem{R^I_a E_{aa^n}} = \one$.
% \end{proof}

\begin{lemma} \label{lem:conn-triv}
If a language $L$ is connected, then for any $u \in \Sigma^+$, 
$R^I_u (D^I_u L) \subseteq \one$.
\end{lemma}

\begin{proof}
  Because $L$ is connected, if $w \in D^I_u L$, then $a D b$ for some
  $a \in u$ and $b \in w$. For such $w$ to also be in
  $R^I_u (D^I_u L)$, we also need that $w I u$. This is only possible
  if $w = \eps$.
\end{proof}

\begin{lemma} \label{lem:triv-synt}
  For any $E$, if $\sem{E} \subseteq \one$, then either $E \doteq 0$ or
  $E \doteq 1$ (using the equations involving $0$ and $1$ only (e.g., 
  $0 + 1 \doteq 1$ and $0^* \doteq 1$ etc.) and
  that $0$ is zero).
\end{lemma}

\begin{lemma} \label{lem:conn-triv-synt-cor} For any $E$, $E'$
  and $u \in \Sigma^+$, if $\sem{E}$ is connected and
  $E \to^{I*} (u, E')$, then $R^I_u E' \doteq 0$ or
  $R^I_u E' \doteq 1$ (using the equations involving $0$ and $1$ only
  and that $0$ is zero).
\end{lemma}

\begin{proof}
  From $E \to^{I*} (u, E')$ by Proposition~\ref{prop:antimirov-sem-deriv},
  $\sem{E'} \subseteq D^I_{u} \sem{E}$. Hence by
  Lemma~\ref{lem:conn-triv}, we get
  $\sem{R^I_{u} E'} = R^I_{u} \sem{E'} \subseteq R^I_{u} (D^I_{u}
  \sem{E}) \subseteq \one$.
  By Lemma~\ref{lem:triv-synt}, $R^I_u E' \doteq 0$ or
  $R^I_u E' \doteq 1$.
\end{proof}

\begin{lemma} \label{lem:star-naive}
  For any $E$, $E'$ and $u \in \Sigma^*$, if $E^* \to^{I*} (u, E')$, then
  there exist $n \in \Nat$, $E_1, \ldots, E_n$,
  $\emptyset \subset X_0, \ldots, X_{n-1} \subseteq \Sigma$ and
  $u_1, \ldots, u_n \in \Sigma^+$ such that
  \[
  E' \doteq (R^I_{X_0} E)^* (R^I_{X_1} E_1) (R^I_{X_1} E)^* \ldots
  (R^I_{X_{n-1}} E_{n-1}) (R^I_{X_{n-1}} E)^* E_{n} E^*
  \]
  where $X_{i-1} \supseteq X_i \cup \Sigma(u_i)$ and
  $E \to^{I*} (u_i, E_i)$ for all $i$ (using only associativity of
  $\cdot$).
\end{lemma}

\begin{lemma} \label{lem:star-conseq-equal}
  For any $E$, $E'$ and $u \in \Sigma^*$, if $\sem{E}$ is connected,
  $E^* \to^{I*} (u, E')$ and, for the development of $E'$ from the
  previous lemma, we have $X_{i-1} = X_i$ for some $i$, then
  $R^I_{X_i} E_i \doteq 0$ or $R^I_{X_i} E_i \doteq 1$ (using the equations involving $0$ and $1$ only, 
  that $0$ is zero).
\end{lemma}

\begin{proof}
  We have $\Sigma (u_i) \subseteq X_{i-1} = X_i$.  From
  $E \to^{I*} (u_i, E_i)$, by Lemma~\ref{lem:conn-triv-synt-cor}
  either $R^I_{u_i} E_i \doteq 0$ or $R^I_{u_i} E_i \doteq 1$.  Therefore also
  $R^I_{X_i} E_i \doteq 0$ or $R^I_{X_i} E_i \doteq 1$.
\end{proof}

%  From
%   $E^* \to^{I*} (u, E')$ by WHATEVER,
%   $\sem{E_i} \subseteq D^I_{u_i} \sem{E}$.  Therefore
%   $\sem{R^I_{X_i} E_i} = R^I_{X_i} \sem{E_i} \subseteq R^I_{X_i}
%   (D^I_{u_i} \sem{E}) = R^I_{X_i \setminus \Sigma{u_i}} (R^I_{u_i}
%   (D^I_{u_i} \sem{E})) \subeteq \one$ by Lemma WHATEVER.
% \end{proof}

\begin{lemma} \label{lem:starconn-at-most-sigma} For any $E$,
  $E'$ and $u \in \Sigma^*$, if $\sem{E}$ is connected and
  $E^* \to^{I*} (u, E')$, then there exist $n \leq |\Sigma|$,
  $E_1, \ldots, E_n$ and $\emptyset \subset X_0, \ldots, X_{n-1} \subseteq \Sigma$ such that 
  \[
  E' \doteq (R^I_{X_0} E)^* (R^I_{X_1} E_1) (R^I_{X_1} E)^* \ldots
  (R^I_{X_{n-1}} E_{n-1}) (R^I_{X_{n-1}} E)^* E_{n} E^*
  \]
  and $X_{i-1} \supset X_i$ for all $i$
  (using, in addition to the equations mentioned in the lemmata above,
  unitality of $1$ and the equation $F^* \cdot F^* \doteq F^*$).
\end{lemma}

\begin{proof}
  From Lemmata~\ref{lem:star-naive}, \ref{lem:star-conseq-equal} noting 
    that at most $| \Sigma | - 1$ of the inclusions
    $X_{i-1} \supseteq X_i$ can be proper.
\end{proof}

\begin{definition} We define functions
  $(\_)^{\to+}, (\_)^{\to*} : \RE \to \P \RE$ by
\[
\small
\begin{array}{rcl}
  a^{\to+}       & \eqdf & \{ 1 \} \\ 
  0^{\to+}       & \eqdf & \emptyset \\
  (E + F)^{\to+} & \eqdf & E^{\to+} \cup F^{\to+} \\
  1^{\to+}       & \eqdf & \emptyset \\
  (EF)^{\to+}    & \eqdf & E^{\to+} \cdot \{ F \} \cup 
                           \bigcup \{ R^I_X (E^{\to*}) \cdot F^{\to+} \mid 
                           \emptyset \subset X \subseteq \Sigma \} \\
  (E^*)^{\to+}   & \eqdf & \{ (R^I_{X_0} E)^* (R^I_{X_1} E_1) \ldots 
                           (R^I_{X_{n-1}} E_{n-1}) (R^I_{X_{n-1}} E)^* E_{n} E^* \mid \\
               &       & \ssp  %\mid 
                           n > 0, \emptyset \subset X_i \subseteq \Sigma, X_{i-1} \supseteq X_i, E_i \in E^{\to+} \} \\[2ex]
  E^{\to*}     & \eqdf & \{ E \} \cup E^{\to+}
\end{array}
\]
\end{definition}

\begin{proposition} \label{prop:antimirov-successors} For any $E$,
  $E'$ and $u \in \Sigma^*$, if $E \to^{I*} (u, E')$, then there
  exists $E''$ such that $E' \doteq E''$ and $E'' \in E^{\to*}$ (using
  only the equations mentioned in the above lemmata).
\end{proposition}

\begin{proposition} \label{prop:starconn-antimirov-finite} If $E$ is
  star-connected, then a suitable sound quotient of the state set
  $\{ E' \mid \exists u \in \Sigma^*.\ E \to^{I*} (u, E') \}$ of the
  Antimirov automaton for $E$ (accepting $\sem{E}^I$) is
  finite.
\end{proposition}

\begin{proof}
  By Lemma~\ref{lem:starconn-at-most-sigma}, for a star-connected
  expression $E$, we only need to consider $n \le | \Sigma |$ in the
  definition of $(E^*)^{\to+}$ for
  Proposition~\ref{prop:antimirov-successors} to hold. This
  restriction makes the set $E^{\to*}$ finite.
\end{proof}

\section{Uniform Scattering Rank of a Language} \label{sec:uniform-scattering-rank}

We proceed to defining the notion of uniform scattering rank of a
language and show that star-connected expressions define languages
with uniform scattering rank.

\subsection{Scattering Rank vs.\ Uniform Scattering Rank}

The notion of scattering rank of a language (a.k.a.\ distribution rank, $k$-block
testability) was introduced by Hashiguchi~\cite{Hashiguchi91}. 

\begin{definition} \label{def:rank}
  A language $L$ has \emph{($I$-scattering) rank} at most $N$ if \\
  $\forall u, v .\, uv \in [L]^I \imp \exists z \in L.\, 
  u \lsim_N z \rsim v$.
\end{definition}

We say that $L$ has rank if it has rank at most $N$ for some
$N \in \Nat$. If it does, then, for the least such $N$, we say that
$L$ has rank $N$.

The only languages with rank 0 are $\emptyset$ and $\one$. If a nontrivial language
$L$ is closed, it has rank 1: for any $uv \in [L]^I$, we have also
have $uv \in L$ and $u \lhd uv \rhd \eps, v$.

% \begin{proposition} (Hashiguchi) Let $(\Sigma , I)$ be a concurrency
%   alphabet and $L$ be a regular language. If $L$ has 
%   rank at most $N$, then $[L]^I$ is recognizable.
% \end{proposition}

% \begin{example}
%   Let $\Sigma = \{ a , b \}$ and $a I b$. The language $\sem{(ab)^*}$
%   does not have a rank.
% \end{example}

Having rank  is a sufficient condition for regularity of the
trace closure of a regular language. But it is not a necessary
condition.

\begin{proposition}[Hashiguchi \cite{Hashiguchi91}]
  (cf.~\cite[Prop.~6.3.2]{Ochmanski95}) \label{prop:hashiguchi} If a
  regular language $L$ has rank, then $[L]^I$ is regular.
\end{proposition}

\begin{proposition} \label{prop:closure-reg-no-rank}
  There exist regular languages $L$ such that $[L]^I$
  is regular but $L$ is without a rank.
\end{proposition}

\begin{proof}
  Consider $\Sigma \eqdf \{a,b\}$, $a I b$. The regular language
  $L \eqdf \sem{(ab)^*(a^*+b^*)}$ is without a rank, since, for any
  $n$, we have $(ab)^n \in L$ and $a^nb^n \in [L]^I$ while the smallest
  $N$ such that $a^n \lsim_N (ab)^n \rsim b^n$ is $n$. Nonetheless,
  $[L]^I = \Sigma^* = \sem{(a+b)^*}$ is regular.
\end{proof}

We wanted to show that a truncation of the refined Antimirov automaton
(which we define in Section~\ref{sec:derivative-uniform-rank}) is
finite for regexps whose language has rank.  But it turns out, as we
shall see, that rank does not quite work for this.  For this reason,
we introduce a stronger notion that we call uniform scattering rank.

\begin{definition}
  A language $L$ has \emph{uniform ($I$-scattering) rank} at most $N$ if \\
  $\forall w \in [L]^I .\ \exists z \in L .\ \forall u, v.\, w = uv \imp
  u \lsim_N z \rsim v$.
\end{definition}

The difference between the two definitions is that, in the uniform
case, the choice of $z$ depends only on $w$ whereas, in the
non-uniform case, it depends on the particular split of $w$ as $w=uv$,
i.e., for every such split of $w$ we may choose a different $z$.

\begin{lemma}
  If $L$ has uniform rank at most $N$, then $L$ has
  rank at most $N$.
\end{lemma}

The converse of the above lemma does not hold---there are languages
with uniform rank greater than rank. Furthermore, there are languages
that have rank but no uniform rank.

\begin{proposition} \label{prop:fin-rank-no-urank}
  Let $\Sigma \eqdf \{ a , b , c \}$, $a I b$ and 
  $E \eqdf a^*b^*c(ab)^*(a^* + b^*) + (ab)^*(a^* + b^*)ca^*b^*$.
  \begin{enumerate}
  \item The language $\sem{E}$ has rank 2.
  \item The language $\sem{E}$ has no uniform rank.
  \end{enumerate}
\end{proposition}

\begin{proof} 
  Note that $c$ behaves like a separator---although $a$ and $b$ are
  independent, neither $a$ nor $b$ commutes with $c$.
  It can be seen that words in $\sem{E}^I$ are of the form $w_lcw_r$ where
  $w_l$ and $w_r$ consist of some number of $a$'s and $b$'s, i.e.,
  $\sem{E}^I = \sem{(a + b)^*c(a + b)^*}$. 
  \begin{enumerate}
  \item Let $uv \in \sem{E}^I$. We have to find $u_1, u_2$ and
    $v_0, v_1, v_2$ so that $u_1u_2 \sim^I u$, $v_0v_1v_2 \sim^I v$,
    $v_0 I u_1$, $v_0v_1 I u_2$ and $v_0u_1v_1u_2v_2 \in \sem{E}$. 
    There are two cases to consider: either $c$ is in the suffix $v$
    or it is in the prefix $u$.
    \begin{itemize}
    \item Case $c \in v$: We have that $u$ consists of only $a$'s and
      $b$'s. Let $x, y \in \Sigma^*$ be such that $v = xcy$. Set
      $u_1 \eqdf \pi_a(u)$, $u_2 \eqdf \pi_b(u)$, $v_0 \eqdf \eps$ and
      $v_1 \eqdf \pi_a(x)$. Let $k \eqdf | y |_a$, $l \eqdf | y |_b$ and
      $m \eqdf \min (k , l)$. Set
      $v_2 \eqdf \pi_b(x) c (ab)^{m} a^{k - m} b^{l - m}$. We have that
      $u_1u_2 = \pi_a(u) \pi_b(u) \sim^I u$. Since
      $\pi_a(x) \pi_b(x) \sim^I x$ and
      $(ab)^ma^{k - m}b^{l - m} \sim^I y$, we also have
      $v_0v_1v_2 = \eps \pi_a(x) \pi_b(x) c (ab)^ma^{k - m}b^{l - m}
      \sim^I xcy = v$.
      Also, $\eps I \pi_a(u)$ and $\eps \pi_a(x) I \pi_b(u)$. Finally,
      $v_0u_1v_1u_2v_2 = \eps \pi_a(u) \pi_a(x) \pi_b(u) \pi_b(x) c
      (ab)^{m} a^{k - m} b^{l - m} \in \sem{a^*b^*c(ab)^*(a^* +
        b^*)}$.

    \item Case $c \in u$: Similar to the previous case. Let $x$ and $y$ be
      such that $u = xcy$. Let $k \eqdf | x |_a$, $l \eqdf | x |_b$ and
      $m \eqdf \min (k , l)$. Set
      $u_1 \eqdf (ab)^ma^{k - m}b^{l - m}c\pi_a(y)$, $u_2 \eqdf \pi_b(y)$,
      $v_0 \eqdf \eps$, $v_1 \eqdf \pi_a(v)$ and $v_2 \eqdf \pi_b(v)$. In this
      case we have
      $v_0u_1v_1u_2v_2 \in \sem{(ab)^*(a^* + b^*)ca^*b^*}$.
    \end{itemize}

  \item Assume that $\sem{E}$ has uniform rank at most $N$. Take
    $w \eqdf a^{N + 1}b^{N + 1}ca^{N + 1}b^{N + 1} \in \sem{E}^I$. By our
    assumption, there is $z \in \sem{E}$ such that, for all $u$ and
    $v$, if $w = uv$, then $u \lsim_N z \rsim v$.

    Take $u \eqdf a^{N + 1}$ and $v \eqdf b^{N + 1}ca^{N + 1}b^{N + 1}$. Thus,
    for some $n \le N$, $z = v_0u_1v_1 \ldots u_nv_n$ and
    $u = a^{N + 1} \sim^I u_1 \ldots u_n$. Since we have $N + 1$
    letters $a$ to divide into $n \leq N$ words, at
    least one $u_i$ must consist of more than one $a$ and thus $z$ must
    contain two consecutive $a$'s that are before $c$.
    
    Take $u' \eqdf a^{N + 1}b^{N + 1}ca^{N + 1}$ and $v' \eqdf b^{N + 1}$.
    Again, for some $n \le N$, $z = v'_0u'_1v'_1 \ldots u'_nv'_n$ and
    $v' = b^{N + 1} \sim^I v'_0 \ldots v'_n$. Note that $c$ must be in
    one of the $u'_i$'s and thus for all $j < i$ it must be that
    $v'_j = \eps$. Hence $v'_0 = \eps$ and we have $N + 1$ letters $b$
    to divide into $n \leq N$ words and thus at least one $v'_i$
    consists of more than one $b$. This means that $z$ must contain
    two consecutive $b$'s that are after $c$.
    
    The only words in $\sem{E}$ equivalent to $w$ are
    $a^{N + 1}b^{N + 1}c(ab)^{N + 1}$ and
    $(ab)^{N + 1}ca^{N + 1}b^{N + 1}$. Neither of these has at least
    two consecutive $a$'s before $c$ as well as at least two
    consecutive $b$'s after $c$, so neither qualifies as
    $z$. Contradiction. \qedhere

  % \item Let $N \in \Nat$. We show that $\sem{E}$ does not have uniform
  %   rank at most $N$. Take
  %   $w = a^{N + 1}b^{N + 1}ca^{N + 1}b^{N + 1} \in \sem{E}^I$. The
  %   words in $\sem{E}$ equivalent to $w$ are
  %   $a^{N + 1}b^{N + 1}c(ab)^{N + 1}$ and
  %   $(ab)^{N + 1}ca^{N + 1}b^{N + 1}$. Thus we have two cases to
  %   consider.
  %   \begin{itemize}
  %   \item $z = a^{N + 1}b^{N + 1}c(ab)^{N + 1}$: Let
  %     $u = a^{N + 1}b^{N + 1}ca^{N + 1}$ and $v = b^{N + 1}$. Assume
  %     that $\sem{E}$ has uniform rank at most $N$ and let
  %     $u_1 , \ldots , u_n$ and $v_0 , \ldots , v_n$ be the suitable
  %     words such that $v \sim^I v_0 \ldots v_n$ and
  %     $u \sim^I u_1 \ldots u_n$ and $z = v_0u_1v_1 \ldots u_Nv_N$.
  %     Since $c$ is in one of the $u_i$, it must be that $v_0 = \eps$.
  %     Since we have $N + 1$ $b$'s to distribute into $N$ words, it
  %     must be that at least one $v_j$ consists of more than one
  %     $b$. This cannot be since the suffix of $z$ beginning with $c$
  %     does not contain any consecutive $b$'s.

  %   \item $z = (ab)^{N + 1}ca^{N + 1}b^{N + 1}$: Let $u = a^{N + 1}$
  %     and $v = b^{N + 1}ca^{N + 1}b^{N + 1}$. Similar to the previous
  %     case.
  %   \end{itemize}
  \end{enumerate}
\end{proof}

\subsection{Star-Connected Languages Have Uniform Rank}

% \begin{proposition}
%   Let $L \subseteq \Sigma^*$ such that $L = [L]^I$. \\
%   $L$ is regular iff $L = \sem{E}^I$ for some star-connected
%   expression $E$.
% \end{proposition}

Klunder et al.~\cite{KlunderOS05} established that star-connectedness
is a sufficient condition for a regular language to have rank,
although not a necessary one.

\begin{proposition}[Klunder et
  al.~\cite{KlunderOS05}] \label{prop:starconn-rank} Any
  star-connected language has rank.
\end{proposition}

\begin{proof}
  The language $\{a\}$ has rank $1$. The languages $\emptyset$ and
  $\one$ have rank $0$. If two languages $L_1$ and $L_2$ have ranks at
  most $N_1$ resp.\ $N_2$, then $L_1 \cup L_2$ has rank at most
  $\max(N_1, N_2)$ and $L_1 \cdot L_2$ has rank at most $N_1 + N_2$.
  If a general language $L$ has rank at most $N$, then $L^*$ need not
  have rank. For example, for $\Sigma \eqdf\{a,b\}$, $a I b$, the
  language $\{ab\}$ has rank 1, but $\{ab\}^*$ is without rank. But if
  $L$ is also connected, then $L^*$ turns out to have rank at most
  $(N+1)\cdot |\Sigma|$.  The claim follows by induction on the given
  star-connected expression.
\end{proof}

\begin{proposition} \label{prop:rank-not-starconn} There exist
  regular languages with rank (and also with uniform rank)
  that are not star-connected.
\end{proposition}

\begin{proof}
  Consider $\Sigma \eqdf \{a,b\}$, $a I b$. The language
  $L \eqdf \sem{(aa+ab+ba+bb)^*}$ has rank 1, in fact even uniform
  rank 1, because it is closed. The regular expression
  $(aa+ab+ba+bb)^*$ is clearly not star-connected, since the language
  $\sem{aa+ab+ba+bb}$ contains disconnected words $ab$ and $ba$. But a
  more involved pumping argument also shows that $L$ is not
  star-connected, i.e., that there is no star-connected expression $E$
  such that $L = \sem{E}$.
\end{proof}

We will now show that star-connected languages also have uniform rank,
by refining Klunder et al.'s proof of
Proposition~\ref{prop:starconn-rank}, especially the case
of the Kleene star.

% It can be seen that if $L_1$ and $L_2$
% have uniform  rank at most $N_1$ and $N_2$, then
% $L_1 \cup L_2$ has uniform  rank at most $\max(N_1, N_2)$
% and $L_1 \cdot L_2$ has uniform  rank at most $N_1 + N_2$.
% If a general $L$ has uniform  rank at most $N$, then 
% %$L^n$ has uniform rank at most $n \cdot N$ if $L$ has uniform rank at most $N$,% so naively $L^*$ need not have a uniform rank, and this may be the case indeed
% $L^*$ need not have uniform rank. For example, for
% $\Sigma \eqdf\{a,b\}$, $a I b$, the language $\{ab\}$ has uniform rank
% 1, but $\{ab\}^*$ is without rank, so also without uniform rank.  But
% if $L$ is also connected, then $L^*$ has uniform rank at most
% $(N+1)\cdot |\Sigma|$.

Let us analyze the case $L^*$ where $L$ is a connected language. When
$w \in [L^*]^I$, then there exists $z \in L$ such that $w \sim^I z$.
This further means that there exist $n \in \N$ and
$z_1, \ldots, z_n \in L$ such that $z = z_1 \ldots z_n$ where we can
require that all $z_i$ are nonempty. Since $L$ is connected, each
$z_i$ is also connected. If $w = uv$, then there exist
$u_1, \ldots, u_n$ and $v_1, \ldots, v_n$ such that
$u \sim^I u_1 \ldots u_n$, $v \sim^I v_1 \ldots v_n$ and, for every
$i$, $z_i \sim^I u_iv_i$ and, for every $j < i$, $v_j I u_i$. In other
words, $u_i$ is the part of $z_i$ that belongs to $u$ and $v_i$ is the
part that belongs to $v$. In particular, if $u_i = \eps$ (or
$z_i \sim^I v_i$), then all letters of $z_i$ belong to the suffix $v$,
and similarly if $v_i = \eps$ (or $z_i \sim^I u_i$), then all letters
of $z_i$ belong to the prefix $u$. An important observation for us is
that not more than $|\Sigma|$ of the $z_i$ can be two-colored in the
sense that both $u_i \neq \eps$ and $v_i \neq \eps$.

\begin{lemma} \label{lem:at-most-sigma-uv} 
  Let $w, u, v, z_1 , \ldots , z_n$ be words such that $w = uv$,
  $w \sim^I z_1 \ldots z_n$ and each $z_i$ is nonempty and connected. Let
  $u_1 , \ldots , u_n, v_1 , \ldots , v_n$ be words such that
  $u \sim^I u_1 \ldots u_n$, $v \sim^I v_1 \ldots v_n$, for all $i$,
  $z_i \sim^I u_i v_i$, and, for all $j < i$,
  $v_j I u_i$.
  For at most $| \Sigma |$ of the words $z_i$, it can be that both $u_i \neq \eps$
  and $v_i \neq \eps$.
\end{lemma}
\begin{proof}
  If, for some $i$, we have that $u_i \neq \eps$ and $v_i \neq \eps$,
  then, since $z_i \sim^I u_iv_i$ is connected, there must exist
  letters $a$ and $b$ such that $a \in u_i$, $b \in v_i$ and $a D b$.
  Since %$v_1 \ldots v_{i - 1} I u_i$ 
  $v_i I u_{i + 1} \ldots u_{n}$, we have
  %$b \not \in v_1 \ldots v_{i - 1}$ and
  $a \not \in u_{i + 1} \ldots u_{n}$.
  This means that, if there are $k$ words $z_i$ such that
  $u_i \neq \eps$ and $v_i \neq \eps$, then these words together must
  contain at least $k$ distinct letters.
\end{proof}

Should it happen for some $i$ that $z_i$ and $z_{i + 1}$ are
completely from the prefix $u$ (in the sense that
$v_i = v_{i+1} = \eps$, i.e., $z_i \sim^I u_i$ and
$z_{i+1} \sim^I u_{i+1}$), then $z_i$ and $z_{i+1}$ belong to the same
block of $u$ in the scattering $u \lsim z \rsim v$, which can
potentially help keeping the uniform rank of $L^*$ low. The same holds
for $z_i$ and $z_{i+1}$ that are completely from the suffix $v$: they
belong to the same block of $v$. Having words $z_i$ that are
completely from $u$ interspersed with other types of words $z_i$ (for
example, having all odd-numbered $z_i$ completely from $u$ and all
even-numbered are completely from $v$), in contrast, is not
helpful. It could thus be useful to be able to choose $z$, $n$ and
$z_1 \ldots, z_n$ in such a way that as many as possible of the $z_i$
that are completely from $u$ are adjacent in $z$ for all splits of $w$
of as $w = uv$.

For example, take $\Sigma \eqdf \{ a , b \}$, $a I b$ and
$L \eqdf \Sigma = \sem{a+ b}$. For $w \eqdf a^mb^m \in [L^*]^I$, we
could build $z = a^mb^m \in L^*$ from $n \eqdf 2m$, $z_i \eqdf a$ for
$1 \le i \le m$ and $z_i \eqdf b$ for $m + 1 \le i \le 2m$.  Another
option is to construct $z = (ab)^m$ from $n \eqdf 2m$,
$z_{2i - 1} \eqdf a$ and $z_{2i} \eqdf b$ for $1 \le i \le m$. In the
first case, the letters from $u$ stay together in $z$ for all prefixes
$u$ of $w$ (as $w = z$). In the second case, they can be interleaved
with the letters from $v$ (in the most extreme case $u \eqdf a^m$ and
$v \eqdf b^m$, the words $u$ and $v$ get scattered into $m$ resp.\
$m+1$ blocks in $z$). Note that, when $z_i \sim^I v_i$ and
$z_{i+1} \sim^I u_{i+1}$, then $z_i I z_{i + 1}$.  The next lemma says
that, for given $z$, $n$, $z_1, \ldots, z_n$, the sequence of words
$z_1, \ldots, z_n$ can be permuted into $z'_1, \ldots, z'_n$ with
$z_1 \ldots z_n \sim^I z'_1 \ldots z'_n$ so that, if $z'_i$ is
completely from $u$, then $z'_{i-1}$ is not completely from $v$. As
$z_i \in L$ for all $i$, it is of course the case that
$z' \eqdf z'_1 \ldots z'_n \in L^*$, so $z'$ is as good a witness of
$w \in [L^*]^I$ as $z$.

\begin{lemma} \label{lem:permutation} Let
  $w, u, v, z_1 , \ldots , z_n$ be words such that $w = uv$,
  $w \sim^I z_1 \ldots z_n$, and each $z_i$ is nonempty and
  connected.
  There exists a permutation $\sigma' = z'_1, \ldots, z'_n$ of
  $\sigma \eqdf z_1, \ldots, z_n$ with the following properties:
\begin{enumerate}
\item $z_1 \ldots z_n \sim^I z'_1 \ldots z'_n$;
\item for
  any $u', v'$ such that $u = u'v'$, and for any
  $u_1, \ldots, u_n, v_1, \ldots v_n$ such that
  $u' \sim^I u_1 \ldots u_n$, $v'v \sim^I v_1 \ldots v_n$, for all $i$,
  $z'_i \sim^I u_i v_i$, and, for all $j < i$, $v_j I u_i$, we have: if
  $v_i = \eps$, then 
  $u_{i - 1} \neq \eps$ unless $i =1$.
\end{enumerate}
\end{lemma}
\begin{proof}
  By induction on $u$.
  \begin{itemize}
  \item Case $\eps$: The identical permutation $\sigma' \eqdf \sigma$
    has property 1 trivially. It also enjoys property 2 since
    $\eps = u'v'$ implies $u' = v' = \eps$, and, for all $i$, we have
    $z_i \sim^I u_iv_i$, $z_i \neq \eps$, $u_i = \eps$ and hence
    $v_i \neq \eps$.

  \item Case $ua$: By induction hypothesis, we have a permutation
    $\sigma' = z'_1, \ldots, z'_n$ of $\sigma$ which has property 1
    and and also has property 2 for all prefixes $u'$ of $w$ up to
    $u$. Now consider the case where $u' \eqdf ua$ and
    $v' \eqdf \eps$. This particular $a$ is in one of the
    $z'_i$, say $z'_m$. The only difference with the
    case $u' \eqdf u$ and $v' \eqdf a$ is that this $a$ is now in the
    $u'$ part of $z'_m$ and no longer in the $v'v$ part. Let
    us also note that every nonempty $u_i$ has remained
    nonempty.

    If $v_m \neq \eps$, then the empty $v_i$
    are exactly the same as in the case $u' \eqdf u$. Hence $\sigma'$
    has property 2 also for the prefix $u' \eqdf ua$.

    If $v_m = \eps$, but $m = 1$ or
    $u_{m-1} \neq \eps$, then $\sigma'$ has property 2
    also for the prefix $u' \eqdf ua$.

    In the critical case $v_m = \eps$, $m \neq 1$ and
    $u_{m-1} = \eps$, we construct a new permutation
    $\sigma'' \eqdf z''_1, \ldots, z''_n$ from $\sigma'$ by moving the
    words $z'_m, \ldots, z'_l$ (where $l$ is the largest such that
    $v_m \ldots v_l = \eps$) in front of of $z'_k, \ldots, z'_{m-1}$
    (where $k$ is the smallest such that
    $u_k \ldots u_{m-1} = \eps$).  Moving all these words rather
    than just $z_m$ alone ensures that the new permutation $\sigma''$
    has property 2 also for all prefixes $u'$ up to $u' = u$ and not
    just only for the prefix $u' = ua$.  The new permutation
    $\sigma''$ also has property 1: indeed, we have
    $z_1 \ldots z_n \sim^I z'_1 \ldots z'_n \sim^I z''_1 \ldots
    z''_n$ 
    as
    $z'_k \ldots z'_{m-1} \sim^I v_k \ldots v_{m-1} ~I~ u_m \ldots
    u_l \sim^I z'_m \ldots
    z'_l$.
%
    % since $u_{i - 1} = \eps$ implies $z_{i - 1} \sim^I v_{i - 1}$,
    % $v_i = \eps$ implies $z_i \sim^I u_i$ and, by assumption,
    % $v_{i - 1} I u_i$. 
  \qedhere
  \end{itemize}
\end{proof}

\begin{proposition} 
\label{prop:starconn-urank}
  If $E$ is star-connected, then the language $\sem{E}$ has uniform rank.
\end{proposition}

\begin{proof}
  By induction on $E$. We only look at the case $E^*$.
\begin{itemize}

\item Case $E^*$: From the assumption we have that $E$ is
  star-connected and $\sem{E}$ is connected. By induction hypothesis
  $\sem{E}$ has uniform rank at most $N$ for some $N \in \Nat$. We show that $\sem{E^*}$
  has uniform rank at most $(| \Sigma | + 1)N$. 

  Let $w \in \sem{E^*}^I$. Then there exist unique $n$ and
  $w_1, \ldots , w_n$ such that
  $w \in w_1 \cdot^I \ldots \cdot^I w_n$, $w_i \in \sem{E}^I$, and we
  can also require that $w_i \neq \eps$. By $\sem{E}$ having
  uniform rank at most $N$, for every $i$, there exists a nonempty
  word $z_i \in \sem{E}$ such that, for any split of $w_i$ as $u_iv_i$,
  we have $u_i \lsim_N z_i \rsim v_i$. By connectedness of $\sem{E}$,
  all $z_i$ are connected.

  We take $z^\dg \eqdf z^\dg_1 \ldots z^\dg_n$ where
  $\sigma^\dg \eqdf z^\dg_1, \ldots, z^\dg_n$ is the permutation of
  $\sigma = z_1, \ldots, z_n$ obtained by Lemma~\ref{lem:permutation}
  for $u \eqdf w$ and $v \eqdf \eps$, i.e., for the specific split of
  $w$ as $w\eps$. By  Lemma~\ref{lem:permutation}(1), we have
  $w \sim^I w_1 \ldots w_n \sim^I z_1 \ldots z_n \sim^I z^\dg_1 \ldots
  z^\dg_n = z^\dg$.
  We let $w^\dg_1, \ldots, w^\dg_n$ be the corresponding permutation
  of $w_1, \ldots, w_n$, so we also have
  $w \sim^I w^\dg_0 \ldots w^\dg_n$ and $w^\dg_i \sim^I z^\dg_i$ for
  all $i$.

  We will now show that, for any split of $w$ as $w = uv$, we have
  $u \lsim_{(|\Sigma|+1)N} z^\dg \rsim v$.

  Let $w = uv$ be any split of $w$. There exist unique
  $u_1, \ldots, u_n, v_1, \ldots v_n$ such $u \sim^I u_1\ldots u _n$,
  $v \sim^I v_1\ldots v_n$, for all $i$, $w^\dg_i = u_i v_i$, and for
  all $j< i$, $v_j I u_i$. They give us
  $u_1 \ldots u_n \lhd z^\dg \rhd v_1 \ldots v_n$ and thus
  $u \lsim z^\dg \rsim v$. By Lemma~\ref{lem:permutation}(2) for
  $u' \eqdf u$, $v' \eqdf \eps$, we have that $v_i = \eps$
  implies $u_{i-1} \neq \eps$ unless $i = 1$.  By
  Lemma~\ref{lem:at-most-sigma-uv}, there can be at most $| \Sigma |$
  words $z^\dg_i$ such that both $u_i \neq \eps$ and $v_i \neq \eps$.
  
  Each of these two-colored $z^\dg_i$ contributes at most $N$ to the
  degree of $u \lsim z^\dg \rsim v$, so altogether they contribute at
  most $| \Sigma | N$.

  Between any two-colored $z^\dg_i$ and also before the first and
  after the last one of them, there are some $z^\dg_i$ completely from
  $u$ followed by some $z^\dg_i$ completely from $v$.  Each such
  sequence contributes at most 1 to the degree of
  $u \lsim z^\dg \rsim v$. If there are less than $| \Sigma |$ two-colored
  words, these sequences thus contribute altogether at most $|\Sigma|$
  to the degree.  If there are exactly $| \Sigma |$ two-colored words
  $z^\dg_i$, then the $z^\dg_i$ after the last of them are all
  completely from $v$, so their sequence belongs to the last $v$-block
  generated by the last two-colored $z^\dg_i$ and thus contributes 0.  Again
  altogether these sequences contribute at most $|\Sigma|$.

  Altogether the degree of $u \lhd z^\dg \rhd v$ is thus at most
  $(| \Sigma | + 1 )N$. \qedhere
%%
%
% These are the only $z_{\sigma(i)}$ with letters
%   from both $u$ and $v$ and these contribute at most $| \Sigma | N$ to
%   the uniform rank.
%
%   Note that, by Lemma~\ref{lem:permutation}, the $z_{\sigma(i)}$
%   completely from $u$ can be at the beginning
%   ($z_{\sigma(1)} \ldots z_{\sigma(k)}$ for some $k$) and immediately
%   after a $z_{\sigma(l)}$ that is scattered between $u$ and $v$
%   ($z_{\sigma(l + 1)} \ldots z_{\sigma(l + 1 + m)}$ for some $l$ and
%   $m$ such that $u_{\sigma(l)} \neq \eps$ and
%   $v_{\sigma(l)} \neq \eps$). If, for some $k$ and $l$, we have that
%   $z_{\sigma(k)} \ldots z_{\sigma(l)}$ contains $| \Sigma |$  such
%   scattered $z_{\sigma(i)}$, then
%   $v_{\sigma(1)} \ldots v_{\sigma(k - 1)} = \eps$ and
%   $u_{\sigma(l + 1)} \ldots u_{\sigma(n)} = \eps$. This gives us that
%   $u \lsim_{(| \Sigma | + 1 ) N} z \rsim v$
%  .
\end{itemize}
\end{proof}

%Not every regular language with uniform rank is star-connected (see Proposition \ref{prop:rank-not-star-conn} in Appendix \ref{app:tc}).

\section{Antimirov Reordering Derivative and Uniform Rank}\label{sec:derivative-uniform-rank}

We have seen that the reordering language derivative $D^I_u L$ allows
$u$ to be scattered in a word $z \in L$ as
$u_1, \ldots, u_n \lhd z \rhd v_0, \ldots, v_n$ where
$u \sim^I u_1\ldots u_n$.  We will now consider a version of the
Antimirov reordering derivative operation that delivers lists of
regexps for the possible $v_0, \ldots, v_n$ rather than just single
regexps for their concatenations $v_0\ldots v_n$.

% We have seen that the (semantic) reordering derivative does not
% require $u$ to be a prefix of $z$ when deriving $z$ along $u$.
% The letters of $u$ may be scattered in $z$ (according to $I$).
% In this section we make this scattering explicit in the reordering Antimirov derivative.

\subsection{Refined Antimirov Reordering Derivative}

The refined reordering parts-of-derivative of a regexp $E$ along a
letter $a$ are pairs of regexps $E_l, E_r$. For any word
$w = av \in \sem{E}^I$, there must be an equivalent word
$z = v_lav_r \in \sem{E}$. Instead of describing the words $v_lv_r$
obtainable by removing a minimal occurrence of $a$ in a word
$z \in \sem{E}$, the refined parts-of-derivative describe the subwords
$v_l, v_r$ that were to the left and right of this $a$ in $z$: it must
be the case that $v_l \in \sem{E_l}$ and $v_r \in \sem{E_r}$ for one of
the pairs $E_l, E_r$. For a longer word $u$, the refined reordering
derivative operation gives lists of regexps $E_0, \ldots, E_n$ fixing
what the lists of subwords $v_0, \ldots, v_n$ can be in words
$z = v_0u_1v_1 \ldots u_nv_n \in \sem{E}$ equivalent to a given word
$w = uv \in \sem{E}^I$.

% Compared to the definition in Sec.~\ref{sec:antimirov-expr},
% the main difference here is that deriving an expression $E$ along
% $a$ now results in a pair of expressions, $E_l, E_r$. The
% idea is that $a$ was between $E_l$ and $E_r$ (in terms of
% expression multiplication) in $E$. Continuing with this process, deriving along a
% word $u$ would give us a list of expressions $E_0, \ldots, E_n$. The
% expressions $E_i$ are what is left of the expression $E$ and the part
% between these expressions (the commas) is where the letters of $u$
% have been taken from. The $E_i$ correspond to the $v_i$ in the
% definition of the semantic reordering derivative (Definition~\ref{def:reord-deriv}) and the commas
% correspond to the $u_i$.

\begin{definition} The \emph{(unbounded and bounded) refined Antimirov $I$-reordering}
  parts-of-derivatives of a regexp along a letter and a word are given
  by relations
  ${\to^I} \subseteq \RE \times \Sigma \times \RE \times \RE$,
  ${\Rightarrow^{I}} \subseteq \RE^+ \times \Sigma \times \RE^+$, 
  ${\to^{I*}} \subseteq \RE \times \Sigma^* \times \RE^+$, 
  ${\Rightarrow^{I}_N} \subseteq \RE ^{+\leq N+1}\times \Sigma \times \RE^{+\leq N+1}$, and
  ${\to^{I*}_N} \subseteq \RE \times \Sigma^* \times \RE^{+\leq N+1}$
  defined inductively by
\[
\small
\begin{array}{l}
\infer{a \to^I (a; 1, 1)}{
}
\quad
\infer{E + F \to^I (a; E_l, E_r)}{
  E \to^I (a; E_l, E_r)
}
\quad
\infer{E + F \to^I (a; F_l, F_r)}{
  F \to^I (a; F_l, F_r)
}
\\[2ex]
\infer{E F \to^I (a; E_l, E_r F)}{
  E \to^I (a; E_l, E_r)
}
\quad
\infer{E F \to^I (a; (R^I_a E) F_l, F_r)}{
  F \to^I (a; F_l, F_r)
}
\quad
\infer{E^* \to^I (a; (R^I_a E)^* E_l, E_r E^*)}{
  E \to^I (a; E_l, E_r)
}
\end{array}
\]
\[
\small
\begin{array}{l}
\infer{\Gamma, E, \Delta \Rightarrow^{I}_N (a; \R^I_a \Gamma, E_l, E_r, \Delta)}{
  E \to^I (a; E_l, E_r)
  & |\Gamma,\Delta|< N
}
\quad
\infer{\Gamma, E, \Delta \Rightarrow^{I}_N (a; \R^I_a \Gamma, E_r, \Delta)}{
  E \to^I (a; E_l, E_r)
  & 
  E_l \dn
  & 
  |\Gamma| > 0
}
\\[2ex]
\infer{\Gamma, E, \Delta \Rightarrow^{I}_N (a; R^I_a \Gamma, E_l, \Delta)}{
  E \to^I (a; E_l, E_r)
  & 
  E_r \dn
  &
  |\Delta| > 0
}
\quad
\infer{\Gamma, E, \Delta \Rightarrow^{I}_N (a; R^I_a \Gamma, \Delta)}{
  E \to^I (a; E_l, E_r)
  & 
  E_l \dn & E_r \dn & |\Gamma| > 0 & |\Delta| > 0
}
\end{array}
\]
\[
\small
\begin{array}{l}
\infer{E \to^{I*}_N (\eps; E)}{
}
\quad
\infer{E \to^{I*}_N (ua; \Gamma')}{
  E \to^{I*}_N (u; \Gamma)
  &
  \Gamma \Rightarrow^{I}_N (a; \Gamma')
}
% \quad
% \infer{E\to^{I*}_N (ua; \R^I_a \Gamma, E_l, E_r, \Delta)}{
%   E \to^{I*}_N (u; \Gamma, E', \Delta)
%   &
%   E' \to^I (a; E_l, E_r)
%   & |\Gamma,\Delta|< N
% }

% \\[2ex]
% \hspace*{-2mm}
% \infer{E\to^{I*}_N (ua; R^I_a F, \R^I_a \Gamma, E_r, \Delta)}{
%   E \to^{I*}_N (u; F, \Gamma, E', \Delta)
%   &
%   E' \to^I (a; E_l, E_r)
%   & 
%   E_l \dn
% }
% \quad %\\[2ex]
% \infer{E\to^{I*}_N (ua; R^I_a \Gamma, E_l, \Delta, G)}{
%   E \to^{I*}_N (u; \Gamma, E', \Delta, G)
%   &
%   E' \to^I (a; E_l, E_r)
%   & 
%   E_r \dn
% }
% \\[2ex]
% \infer{E\to^{I*}_N (ua; R^I_a F, R^I_a \Gamma, \Delta, G)}{
%   E \to^{I*}_N (u; F, \Gamma, E', \Delta, G)
%   &
%   E' \to^I (a; E_l, E_r)
%   & 
%   E_l \dn & E_r \dn
% }
\end{array}
\]  
By $\RE^{+\leq N+1}$ we mean nonempty lists of regexps of length at
most $N+1$. The relations $\Rightarrow^{I}$ and $\to^{I*}$ are defined
exactly as $\Rightarrow_N^{I}$ and $\to_N^{I*}$ but with the condition
$|\Gamma, \Delta| < N$ of the first rule of $\Rightarrow_N^{I}$
dropped. The operation $R^I_a$ is extended to lists of regexps in the
obvious way.
\end{definition}

We have several rules for deriving a list of regexps along $a$. % so that we can truncate
% the list precisely---only lists of length at most $N + 1$ can be
% derived. 
If $E$ is split into $E_l, E_r$ and neither of them is nullable, then,
in the $N$-bounded case, we require that the given list is shorter
than $N+1$ since the new list will be longer by 1. If one of
$E_l, E_r$ is nullable, not the first resp.\ last in the list and we
choose to drop it, then the new list will be of the same length.  If
both are nullable, not the first resp.\ last and we opt to drop both,
then the new list will be shorter by 1.  They must be droppable under
these conditions to handle the situation when a word $z$ has been
split as $v_0u_1v_1 \ldots u_kv_ku_{k+1} \ldots u_nv_n$ and $v_k$ is
further being split as $v_lav_r$ while $v_l$ or $v_r$ is empty.  If
$k \neq 0$ and $v_l$ is empty, we must join $u_k$ and $a$ into
$u_ka$. If $k \neq n$ and $v_r$ is empty, we must join $a$ and
$u_{k+1}$ into $au_{k+1}$. If $k$ is neither $0$ nor $n$ and both
$v_l$ and $v_r$ are empty, we must join all three of $u_k$, $a$ and
$u_{k+1}$ into $u_kau_{k+1}$.  The length of the new list of regexps
is always at least 2.
%  , then we can drop the $v_i$ and merge $u_i$ and $u_{i+1}$
% together and thus require a smaller number of pieces for $u$ and $v$
% but get the same word. The difference is that by dropping a nullable
% expression $E_i$ we force ourselves to the case where $v_i =
% \eps$.

\begin{proposition} 
\label{prop:holes-main}
For any $E$,
\begin{enumerate} 
\item for any $a \in \Sigma, v_l, v_r \in \Sigma^*$, 
\[
v_l I a \con v_lav_r \in \sem{E}  \iff 
    \exists E_l, E_r.\, E \to^I (a; E_l, E_r) 
   \con v_l \in \sem{E_l} \con v_r \in \sem{E_r};
\]

\item for any $u \in \Sigma^*, n \in \Nat, v_0 \in \Sigma^*,  v_1, \ldots, v_{n-1} \in \Sigma^+, v_n \in \Sigma^*$, 
\[
\begin{array}{l}
\exists z \in \sem{E}, u_1, \ldots, u_n \in \Sigma^+.\, u \sim^I u_1\ldots u_n
\con u_1, \ldots , u_n \lhd z \rhd v_0, \ldots, 
v_n  \\ 
  % \exists z \in \sem{E}.\  u \lsim z \rhd v_0, \ldots,
  % v_n  \\ 
\iff \\
  \exists E_0, \ldots, E_n.\, E \to^{I*} (u; E_0, \ldots, E_n) 
  \con \forall j.\ v_j \in \sem{E_j}.
\end{array}
\]
\end{enumerate}
\end{proposition}

\begin{proof}~ %[Proof of Proposition \ref{prop:holes-main}]~

\begin{enumerate}
\item $\imp$: By induction on $E$.

\begin{itemize}
\item Case $a'$ where $a' \neq a$: $v_l a v_r \in \sem{a'} = \{a'\}$
  is impossible.

\item Case $a$: Suppose $v_l a v_r \in \sem{a} = \{a\}$. Then
  $v_l = v_r = \eps$. We have $a \to^I (a; 1, 1)$ and
  $\eps \in \sem{1}$ as required.

\item Case $0$: $v_l a v_r \in \sem{0} = \emptyset$
  is impossible.

\item Case $E_1 + E_2$: Suppose
  $v_l a v_r \in \sem{E_1 + E_2} = \sem{E_1} \cup \sem{E_2}$ and
  $v_l I a$.  Then $v_l a v_r \in \sem{E_i}$ for one of two possible
  $i$. By IH for $E_i, a, v_l, v_r$, there are $E_l$, $E_r$ such that
  $E_i \to^I (a; E_l, E_r)$, $v_l \in \sem{E_l}$, $v_r \in \sem{E_r}$
  and we also obtain $E_1 + E_2 \to^I (a; E_l, E_r)$.

\item Case $1$: $v_l a v_r \in \sem{1} = \one$
  is impossible.

\item Case $E F$: Suppose
  $v_l a v_r \in \sem{E F} = \sem{E} \cdot \sem{F}$ and $v_l I a$.
  Then $v_l a v_r = x y$ for some $x \in \sem{E}$ and $y \in \sem{F}$.
  Either (i) there exists $v'$ such that $x = v_l a v'$ and
  $v_r = v' y$ or (ii) there exists $v'$ such that $y = v' a v_r$ and
  $v_l = x v'$.

  If (i), then, by IH for $E, a, v_l, v'$, there are $E_l, E_r$ such
  that $E \to^I (a; E_l, E_r)$, $v_l \in \sem{E_l}$,
  $v' \in \sem{E_r}$.  We then also have $E F \to^I (a; E_l, E_r F)$
  and $v_r = v'y \in \sem{E_r F}$.

  If (ii), then $x I a$ and $v' I a$, so $x \in \sem{R^I_a E}$ and, by
  IH for $F, a, v', v_r$, there are $F_l, F_r$ such that
  $F \to^I (a; F_l, F_r)$, $v' \in \sem{F_l}$, $v_r \in \sem{F_r}$.
  We then also have $E F \to^I (a; (R^I_a E) F_l, F_r)$ and
  $v_l = xv' \in \sem{(R^I_a E) F_l}$.

\item Case $E^*$: Suppose $v_l a v_r \in \sem{E^*}$ and $v_l I a$.
  Then $v_l = x v'_l$, $v_r = v'_r y$ for some $x, y \in \sem{E^*}$ and
  $v'_l, v'_r$ such that $v'_l a v'_r \in \sem{E}$. We have $x I a$,
  $v'_l I a$. Hence $x \in \sem{R^I_a E^*}$ and, by IH for
  $E, a, v'_l, v'_r$, we get that there are $E_l$, $E_r$ such that
  $E \to^I (a; E_l, E_r)$ and $v'_l \in \sem{E_l}$,
  $v'_r \in \sem{E_r}$. We also obtain
  $v_l = xv'_l \in \sem{(R^I_a E^*)E_l}$ and
    $v_r = v'_ry \in \sem{E_r E^*}$.

\end{itemize}

$\bwdimp$: By induction on the derivation of $E \to^I (a; E_l, E_r)$.
\begin{itemize}
\item Case $a \to^I (a; 1, 1)$ as an axiom: Suppose
  $v_l, v_r \in \sem{1} = \one$. Then $v_l = v_r = \eps$ and
  we have $\eps a \eps = a \in \sem{a}$ as required.

\item Case $E_1 + E_2 \to^I (a; E_l, E_r)$ inferred from
  $E_i \to^I (a; E_l, E_r)$ where $i$ is 1 or 2: Suppose
  $v_l \in \sem{E_l}$, $v_r \in \sem{E_r}$. We can then apply IH to
  the subderivation, $v_l, v_r$ and obtain that $v_l I a$ and $v_l a v_r \in \sem{E_i}$, which gives us
  also that
  $v_l a v_r \in \sem{E_1} \cup \sem{E_2} = \sem{E_1 + E_2}$.

\item Case $E F \to^I (a; E_l, E_r F)$ inferred from
  $E \to^I (a; E_l, E_r)$: Suppose $v_l \in \sem{E_l}$,
  $v_r \in \sem{E_r F} = \sem{E_r} \cdot \sem{F}$. Then $v_r = x y$
  for some $x \in \sem{E_r}$ and $y \in \sem{F}$. We can then apply IH
  to the subderivation, $v_l, x$ and obtain that $v_l I a$ and
  $v_l a x \in \sem{E}$. As $v_r = xy$, we obtain
  $v_l a v_r = (v_l a x) y \in \sem{E} \cdot \sem{F} = \sem{E F}$.

\item Case $E F \to^I (a, (R_a E) F_l, F_r)$ inferred from
  $F \to^I (a; F_l, F_r)$: Suppose
  $v_l \in \linebreak \sem{(R_aE) F_l} = R_a \sem{E} \cdot \sem{F_l}$,
  $v_r \in \sem{F_r}$. Then $v_l = x y$ for some $x \in R_a \sem{E}$
  and $y \in \sem{F_l}$, which also gives us $x I a$ and
  $x \in \sem{E}$. We can then apply IH to the subderivation, $y, v_r$
  and obtain that $y I a$ and $y a v_r \in \sem{F}$. As $v_l = xy$, we
  get $v_l I a$ and
  $v_l a v_r = x (y a v_r) \in \sem{E} \cdot \sem{F} = \sem{E F}$.

\item Case $E^* \to^I (a; (R_a E^*) E_l, E_r E^*)$ inferred from
  $E \to^I (a; E_l, E_r)$: Suppose
  $v_l \in \sem{(R_aE^*) E_l} = R_a \sem{E^*} \cdot \sem{E_l}$,
  $v_r \in \sem{E_r E^*} = \sem{E_r} \cdot \sem{E^*}$. Then
  $v_l = x y$ for some $x \in R_a \sem{E^*}$ and $y \in \sem{E_l}$,
  which also gives us $x I a$ and $x \in \sem{E^*}$, and $v_r = z w$
  for some $z \in \sem{E_r}$ and $w \in \sem{E^*}$. We can then apply
  IH to the subderivation, $y, z$ and obtain that $y I a$ and
  $y a z \in \sem{E}$. As $v_l = xy$ and $v_r = zw$, we get that
  $v_l I a$ and
  $v_l a v_r = x (y a z) w \in \sem{E^*} \cdot \sem{E} \cdot \sem{E^*} \subseteq
  \sem{E^*}$.

\end{itemize}

\item $\imp$: For any $E$ by induction on $u$.
\begin{itemize}
\item
  Case $\eps$: Suppose $\eps \sim^I u_1 \ldots u_n$ and
  $z = v_0u_1v_1\ldots u_nv_n \in \sem{E}$. Then necessarily $n = 0$, which
  means that we actually have $v_0 \in \sem{E}$. We also have
  $E \to^{I^*} (\eps; E)$ as required.

\item
  Case $ua$: Suppose $ua \sim^I u_1 \ldots u_n$ and
  $\forall i. \forall j < i.\, v_j I u_i$ and
  $z = v_0u_1v_1\ldots u_nv_n \in \sem{E}$. It must be that $n > 0$ and
  there must exist $k$ and $u_l, u_r \in \Sigma^*$ such that
  $u_k = u_l a u_r$, $v_j I a$ for all $j < k$, $a I u_r$,
  $a I u_i$ for all $i> k$ and
  $u \sim^I u_1 \ldots u_{k-1}u_lu_ru_{k+1} \ldots u_n$.

\begin{itemize}
\item  If $u_l = u_r = \eps$, then, as $v_{k-1}av_k I u_i$ for all $i> k$,
  we are entitled to apply IH to $u$, $n-1$,
  $v_0, v_1, \ldots, v_{k-2}, v_{k-1}av_k, v_{k+1}, \ldots, v_n$, $z$,
  $u_1, \ldots, u_{k-1}, u_{k+1}, \ldots, u_n$.  We get
  $E_0, \ldots, E_{k-2}$, $E'$, $E_{k+1}$, \ldots $E_n$ such that
  $E \to^{I*} (u; E_0, \ldots, E_{k-2}, E', E_{k+1}, \ldots, E_n)$ and
  $v_j \in \sem{E_j}$ for all $j < k-1$, $v_{k-1}av_k \in \sem{E'}$,
  $v_j \in \sem{E_j}$ for all $j > k$. As $v_{k-1} I a$, we can apply
  1. to $E'$, $a$, $v_{k-1}$, $v_k$ and get $E_{k-1}, E_k$ such that \linebreak
  $E' \to^I (a; E_{k-1}, E_k)$, $v_{k-1} \in \sem{E_{k-1}}$ and
  $v_k \in \sem{E_k}$. This allows us to infer \linebreak
  $E \to^{I*} (ua; R^I_aE_0, \ldots, R^I_aE_{k-2}, E_{k-1}, E_k, E_{k+1},
  \ldots, E_n)$.
  As $v_j I a$ also for all \linebreak $j < k-1$, we in fact also have
  $v_j \in \sem{R^I_a E_j}$ for all $j < k-1$.

\item  If $u_l \neq \eps$, $u_r = \eps$, we note that $av_k I u_i$ for all
  $i> k$ and apply IH to $u$, $n$,
  $v_0$, $v_1, \ldots, v_{k-1}$, $av_k, v_{k+1}, \ldots, v_n$, $z$,
  $u_1, \ldots, u_{k-1}, u_l, u_{k+1}, \ldots, u_n$.  We get %\linebreak
  $E_0, \ldots, E_{k-1}$, $E'$, $E_{k+1}, \ldots E_n$ such that
  $E \to^{I*} (u; E_0, \ldots, E_{k-1}, E', E_{k+1}, \ldots, E_n)$ and
  $v_j \in \sem{E_j}$ for all $j < k$, $av_k \in \sem{E'}$,
  $v_j \in \sem{E_j}$ for all $j > k$.  As $\eps I a$, we can apply
  1. to $E'$, $a$, $\eps$, $v_k$ and get $E'', E_k$ such that
  $E' \to^I (a; E'', E_k)$, $\eps \in \sem{E''}$, $v_k \in
  \sem{E_k}$. As $\eps \in \sem{E''}$ tells us that $E'' \dn$,
  we can infer
  $E \to^{I*} (ua; R^I_aE_0, \ldots, R^I_aE_{k-1}, E_k, E_{k+1},
  \ldots, E_n)$.
  As $v_j I a$ also for all $j <k$, we in fact also have
  $v_j \in \sem{R^I_a E_j}$ for all $j < k$.

\item The cases $u_l = \eps$, $u_r \neq \eps$ and $u_l \neq \eps$,
  $u_r \neq \eps$ are handled similarly to the previous case.

  % If $u_l = \eps$, $u_r \neq \eps$, we proceed similarly and apply IH to 
  % $E, u, n$,
  % $v_0, v_1, \ldots, v_{k-1}a, v_k, v_{k+1}, \ldots, v_n$,
  % $u_1, \ldots, u_{k-1}, u_r, u_{k+1}, \ldots, u_n$.

  % If $u_l \neq \eps$, $u_r \neq \eps$, IH has to be applied to
  % $E, u, n+1$,
  % $v_0, v_1, \ldots, v_{k-1}, a, v_k, v_{k+1}, \ldots, v_n$,
  % $u_1, \ldots, u_{k-1}, u_l, u_r, u_{k+1}, \ldots, u_n$.
\end{itemize}
\end{itemize}

$\bwdimp$: By induction on the derivation of
$E \to^{I*} (u; E_0, \ldots, E_n)$.
\begin{itemize}
\item Case $E \to^{I*} (\eps; E)$ as an axiom: Suppose that
    $v_0 \in \sem{E}$. We have $\eps \sim \eps$ as well as
    $z = v_0 \in \sem{E}$ directly. 

  \item Case
    $E\to^{I*} (ua; R^I_a E_0, \ldots , R^I_a E_{k-1}, E_l, E_r,
    E_{k+1}, \ldots, E_n)$
    inferred from subderivations  $E \to^{I*} (u; E_0, \ldots, E_n)$ and
    $E_k \to^I (a; E_l, E_r)$: Suppose that $v_0 \in \sem{R^I_a E_0}$,
    \ldots, $v_{k-1} \in \sem{R^I_a E_{k-1}}$, $v_l \in \sem{E_l}$,
    $v_r \in \sem{E_r}$, $v_{k+1} \in \sem{E_{k+1}}$, \ldots,
    $v_n \in \sem{E_n}$, which gives us  $v_0 I a$,
    \ldots, $v_{k-1} I a$, $v_0 \in \sem{E_0}$, \ldots,
    $v_{k-1} \in \sem{E_{k-1}}$. Applying (1.$\bwdimp$) to
    $E_k$, $a$, $v_l$, $v_r$, we learn that $v_l I a$ and
    $v_lav_r \in \sem{E_k}$.  Applying IH to the subderivation,
    $v_0, \ldots, v_{k-1}$, $v_lav_r$, $v_{k+1}, \ldots, v_n$, we obtain
    $z, u_1, \ldots, u_n \in \Sigma^+$ such that
    $u \sim^I u_1\ldots u_n$,
    $\forall i.\, \forall j < i.\ v_j I u_i$,
    $\forall i > k.\, v_lav_r I u_i$ and
    $z = v_0u_1v_1\ldots v_{k-1}u_k(v_lav_r)u_{k+1}v_{k+1} \ldots u_nv_n
    \in \sem{E}$.
    Clearly $ua \sim^I u_1 \ldots u_k a u_{k+1} \ldots u_n$ and 
    $z = v_0u_1v_1\ldots v_{k-1}u_kv_lav_ru_{k+1}v_{k+1}\ldots u_nv_n \in
    \sem{E}$.

  \item Case
    $E\to^{I*} (ua; R^I_a E_0, \ldots, R^I_a E_{k-1}, E_r, E_{k+1},
    \ldots, E_n)$
    inferred from subderivations $E \to^{I*} (u; E_0, \ldots, E_n)$ and
    $E_k \to^I (a; E_l, E_r)$ and $E_l \dn$ whereby $k \neq 0$: 
    Suppose that $v_0 \in \sem{R^I_a E_0}$,
    \ldots, $v_{k-1} \in \sem{R^I_a E_{k-1}}$,
    $v_r \in \sem{E_r}$, $v_{k+1} \in \sem{E_{k+1}}$, \ldots,
    $v_n \in \sem{E_n}$, which gives us  $v_0 I a$,
    \ldots, $v_{k-1} I a$, $v_0 \in \sem{E_0}$, \ldots,
    $v_{k-1} \in \sem{E_{k-1}}$. Applying (1.$\bwdimp$) to
    $E$, $a$, $\eps$, $v_r$, we learn that 
    $av_r \in \sem{E_k}$.  Applying IH to the subderivation,
    $v_0, \ldots, v_{k-1}$, $av_r$, $v_{k+1}, \ldots, v_n$, 
    we obtain $z, u_1, \ldots, u_n \in \Sigma^+$ such
    that $u \sim^I u_1\ldots u_n$,
    $\forall i.\, \forall j < i.\ v_j I u_i$, $\forall i > k.\, av_r I u_i$ and
    $z =  v_0u_1v_1\ldots v_{k-1}u_k(av_r)u_{k+1}v_{k+1} \ldots u_nv_n \in \sem{E}$.   Now clearly
    $ua \sim^I u_1 \ldots (u_k a) u_{k+1} \ldots u_n$ and
    $z = v_0u_1v_1\ldots v_{k-1}(u_ka)v_ru_{k+1}v_{k+1}\ldots u_nv_n \in
    \sem{E}$.

\item The two remaining cases are treated similarly to the previous case. \qedhere

  % \item Case
  %   $E \to^{I*} (u, E_0, \ldots, E_{k-1}, E_{k+1}, \ldots, E_n)$
  %   inferred from $ E \to^{I*} (u, E_0, \ldots, E_n)$ and $E_k \dn$
  %   where $0 < k < n$: Suppose $v_0 \in \sem{E_0}$, \ldots,
  %   $v_{k-1} \in \sem{E_{k-1}}$, $v_{k+1} \in \sem{E_{k+1}}$,
  %   $v_n \in \sem{E_n}$. From $E_k \dn$, we have that
  %   $\eps \in \sem{E_k}$.  We apply IH to the subderivation
  %   $v_0, \ldots, v_{k-1}, \eps, v_k, \ldots, v_n$, we get
  %   $u_1, \ldots, u_n$ such that $u \sim^I u_1\ldots u_n$, 
  %   $\forall i.\, \forall j < i, j \neq k.\ v_j I u_i$
  %   and $v_0u_1v_1\ldots v_{k-1}u_k\eps u_{k+1}v_{k+1} \ldots u_nv_n \in \sem{E}$.
  %   This is the same as to say that $u \sim^I u_1\ldots (u_ku_{k+1}) \ldots u_n$,
  %   $\forall i <k.\, \forall j < i.\ v_j I u_i$,
  %   $\forall j< k.\, v_j I (u_ku_{k+1})$,
  %   $\forall i > k.\, j < i+1, j \neg k.\, v_j I u_{i+1}$ and 
  %   $v_0u_1v_1\ldots v_{k-1}(u_ku_{k+1})v_{k+1} \ldots u_nv_n \in \sem{E}$.
   
\end{itemize}

\end{enumerate}
\end{proof}

\begin{proposition} 
\label{prop:holes-cor}
For any $E$,
\begin{enumerate} 
\item for any $a \in \Sigma, v \in \Sigma^*$, the following are equivalent:
\begin{enumerate}
 \item % (i) 
$av \in \sem{E}^I$;
 \item % (ii)
$\exists  v_l, v_r \in \Sigma^*.\, 
 v \sim^I v_l v_r \con v_l I a \con 
 v_l a v_r \in \sem{E}$;
 \item % (iii)  
$\exists  v_l, v_r\in \Sigma^*.\, \\
\ssp v \sim^I v_l v_r \con
\exists E_l, E_r.\, E \to^I (a; E_l, E_r) 
   \con v_l \in \sem{E_l} \con v_r \in \sem{E_r}$;
 \item % (iii)  
$\exists  v_l, v_r \in \Sigma^*.\, \\
\ssp v \in v_l \cdot^I v_r \con 
\exists E_l, E_r.\, E \to^I  (a; E_l, E_r) 
   \con v_l \in \sem{E_l}^I \con v_r \in \sem{E_r}^I$.
\end{enumerate}

\item for any $u, v \in \Sigma^*$, the following are equivalent:
\begin{enumerate}
 \item % (i)
$uv \in \sem{E}^I$;
 \item % (ii) 
$\exists z \in \sem{E}.\  u \lsim z \rsim v$;
 \item % (iii)
$\exists n \in \Nat, v_0 \in \Sigma^*, v_1, \ldots, v_{n-1} \in \Sigma^+, v_n \in \Sigma^*.\,  v \sim^I v_0v_1\ldots v_n \con \\
\ssp \exists E_0, \ldots, E_n.\, E \to^{I*} (u; E_0, \ldots, E_n) 
  \con \forall j.\, v_j \in \sem{E_j}$;
\item % (iv)
$\exists n \in \Nat, v_0 \in \Sigma^*, v_1, \ldots, v_{n-1} \in \Sigma^+, v_n \in \Sigma^*.\, v \in v_0 \cdot^I v_1 \cdot^I \ldots \cdot^I v_n \con \\
\ssp \exists E_0, \ldots, E_n.\, E \to^{I*} (u; E_0, \ldots, E_n) 
  \con \forall j.\, v_j \in \sem{E_j}^I$.
\end{enumerate}

\item for any $u \in \Sigma^*$, \\ \quad\quad
%\[
$
u \in \sem{E}^I \iff 
    (u = \eps \con E \dn) \dis 
    (u \neq \eps \con \exists E_0, E_1.\, E \to^{I*} (u; E_0, E_1) \con E_0 \dn 
     \con E_1 \dn)
$.
%\]
\end{enumerate}
\end{proposition}

\begin{proof}~%[Proof of Proposition \ref{prop:holes-cor}]~
\begin{enumerate}
\item (a) $\iff$ (b) follows from Proposition \ref{prop:scat-uv}.
  (b) $\iff$ (c) follows from Proposition \ref{prop:holes-main}(1). 
(c) $\iff$ (d) follows from Proposition \ref{prop:reord-concat}.

\item (a) $\iff$ (b) follows from Proposition \ref{prop:scat-uv}.
  (b) $\iff$ (c) follows from Proposition \ref{prop:holes-main}(2).
  (c) $\iff$ (d) follows from Proposition \ref{prop:reord-concat}.

\item From (2) for $E, u, \eps$. \qedhere
\end{enumerate}
\end{proof}

\begin{proposition} \label{prop:holes-uniform}
For any $E$, $N \in \Nat$, $u \in \Sigma^*$, $z \in \sem{E}$,
\[
\begin{array}{l}
(\forall u', u''.\, u = u'u'' \imp \exists v.\, u' \lsim_N z \rhd v) \\  
\imp \\
  \exists E_0, \ldots, E_n.\, E \to^{I*}_N (u; E_0, \ldots, E_n) 
  \con \forall j.\ v_j \in \sem{E_j} \\
\quad \textrm{for the unique~}
  n, v_0, \ldots, v_n \textrm{~such that~} 
  u \lsim z \rhd v_0, \ldots, v_n
  \end{array}
\]
\end{proposition}

\begin{proof}
  By replaying the proof of
  Proposition~\ref{prop:holes-main}(2.$\imp$).  In the fourth subcase
  ($u_l \neq \eps$, $u_r \neq \eps$) of the case $ua$ of induction, IH
  for $u' = u$, $u'' = a$ is needed.
\end{proof}

\begin{corollary}
\label{cor:holes-bounded}
For any $E$ such that $\sem{E}$ has uniform rank at most $N$,
\begin{enumerate}
\item for any
  $u, v \in \Sigma^*$, the following are equivalent:
\begin{enumerate}
 \item % (i)
$uv \in \sem{E}^I$;
 \item % (ii) 
%$\exists z. z \in \sem{E} \con u \lsim_N z \rsim v$,
$\exists z \in \sem{E}.\ \forall u',u''.\, u=u' u'' \imp u' \lsim_N z \rsim u''v$;

 \item % (iii)
$\exists n \leq N, v_0 \in \Sigma^*, v_1, \ldots, v_{n-1} \in \Sigma^+, v_n \in \Sigma^*.\,  v \sim^I v_0v_1\ldots v_n \con \\
\ssp \exists E_0, \ldots, E_n.\, E \to^{I*}_N (u; E_0, \ldots, E_n) 
  \con \forall j.\, v_j \in \sem{E_j}$;
\item % (iv)
$\exists n \leq N, v_0 \in \Sigma^*, v_1, \ldots, v_{n-1} \in \Sigma^+, v_n \in \Sigma^*.\, v \in v_0 \cdot^I v_1 \cdot^I \ldots \cdot^I v_n \con \\
\ssp \exists E_0, \ldots, E_n.\, E \to^{I*}_N (u; E_0, \ldots, E_n) 
  \con \forall j.\, v_j \in \sem{E_j}^I$.
\end{enumerate}

\item for any $u \in \Sigma^*$, \\ \quad\quad
%\[
$
u \in \sem{E}^I \iff 
    (u = \eps \con E \dn) \dis 
    (u \neq \eps \con \exists E_0, E_1.\, E \to^{I*}_N (u; E_0, E_1) \con E_0 \dn 
     \con E_1 \dn)
$.
%\]
\end{enumerate}
\end{corollary}

\begin{proof}[Proof of 1]
  (a) $\imp$ (b) is from $E$ having uniform rank at most $N$.  (b)
  $\imp$ (c) follows from Proposition \ref{prop:holes-uniform}.  (c)
  $\imp$ (d) and (d) $\imp$ (a) are those from Proposition
  \ref{prop:holes-cor}.
\end{proof}

\begin{example} % Let $\Sigma \eqdf \{ a, b \}$, $a I b$ and
  % $E \eqdf aa + ab + b$. We write $E_b$ for $R_b E = aa + a0 + 0$.
  We go back to Example~\ref{ex:brz}. Recall that $E \eqdf aa + ab + b$ and $E_b\eqdf R^I_b E = aa + a0 + 0$. Here is one of the refined reordering parts-of-derivatives of $E^*$ along $bb$.
\[
\begin{array}{l}
\small
  \infer{E^* \to^{I*}_2 (bb; E_b^*(a1), 1(E_b^*(a1)), 1E^*)}{
  \infer{E^* \to^{I*}_2 (b; E_b^*(a1), 1E^*)}{
  \infer{E^* \to^{I*}_2 (\eps; E^*)}{} 
& 
  \infer{E^* \Rightarrow^I_2 (b; E_b^*(a1), 1E^*)}{
  \infer{E^* \to^I (b; E^*_b(a1), 1E^*)}{
  \infer{aa + ab + b \to^I (b; a1, 1)}{
  \infer{ab + b \to^I (b; a1, 1)}{
  \infer{ab \to^I (b; a1, 1)}{
  \infer{b \to^I (b; 1, 1)}{}
  }
  }
  }
  }
& 0 < 2}
  }
& 
  \infer{E_b^*(a1), 1E^* \Rightarrow^I_2 (b; E_b^*(a1), 1(E_b^*(a1)), 1E^*)}{
  \infer{1E^* \to^I (b; 1 (E^*_b (a1)), 1E^*)}{
  \infer{E^* \to^I (b; E^*_b (a1), 1E^*)}{
  \infer{aa + ab + b \to^I (b; a1, 1)}{
  \infer{ab + b \to^I (b; a1, 1)}{
  \infer{ab \to^I (b; a1, 1)}{
  \infer{b \to^I (b; 1, 1)}{}
  }
  }
  }
  }
  }
&
  1 < 2}
  }
\end{array}
\]

In this example, we chose $N \eqdf 2$. The regexp
$1(E_b^*(a1)) \doteq (aa)^*a$ is not nullable, so we could not have
dropped it. From here we cannot continue by deriving along a third $b$
by again taking it from the summand $ab$ of $E$ in $1E^*$, as this
would produce another nondroppable $1(E_b^*(a1))$ and make the list
too long (longer than 3). For example, we are not allowed to establish
$w \eqdf bbbaaa \in \sem{E^*}^I$ (by deriving $E^*$ along $w$ and
checking if we can arrive at $E_0, E_1$ with both $E_0, E_1$
nullable), mandated by $z \eqdf ababab \in \sem{E^*}$, but we are
allowed to do so because of $z' \eqdf bbabaa \in \sem{E^*}$.
The word $z$ is not useful since among the splits of $w$ as $w = uv$
there is $u \eqdf bbb$, $v \eqdf aaa$, which splits $z$ as
$u \lsim z \rsim v$ scattering $u$ into 3 blocks as
$z = a\underline{b}a\underline{b}a\underline{b}$ (we underline the
letters from $u$); the full sequence of these corresponding splits of
$z$ is $ababab$, $a\underline{b}abab$,
$a\underline{b}a\underline{b}ab$,
$a\underline{b}a\underline{b}a\underline{b}$,
$\underline{ab}a\underline{b}a\underline{b}$,
$\underline{abab}a\underline{b}$, $\underline{ababab}$.
The word $z'$, on the contrary, is fine because, for every split of
$w$ as $w = uv$, there are at most two blocks of letters from $u$ in
$z'$: $bbabaa$, $\underline{b}babaa$, $\underline{bb}abaa$,
$\underline{bb}a\underline{b}aa$, $\underline{bbab}aa$,
$\underline{bbaba}a$, $\underline{bbabaa}$. The choice $N = 2$
suffices for accepting all of $\sem{E^*}^I$, since $\sem{E^*}$ happens
to have uniform rank 2.
\end{example}

The refined Antimirov reordering parts-of-derivatives of a regexp $E$
give a nondeterministic automaton by
$Q^E \eqdf \{\Gamma \mid \exists u \in \Sigma^*.\, E \to^{I*}(u ;
\Gamma)\}$,
$I^E \eqdf \{E\}$,
$F^E \eqdf \{E \mid E \dn\} \cup \{ E_0, E_1 \in Q^E \mid E_0 \dn \con
E_1 \dn \}$,
$\Gamma \to^E (a; \Gamma') \eqdf \Gamma \Rightarrow^I (a; \Gamma')$
for $\Gamma, \Gamma' \in Q^E$. By Prop.~\ref{prop:holes-cor}, this
automaton accepts $\sem{E}^I$. It is generally not finite as $Q^E$ can
contain states $\Gamma$ of any length.

Given $N \in \Nat$, another automaton is obtained by restricting
$Q^E$, $F^E$ and $\to^E$ to
$Q^E_N \eqdf \{\Gamma \mid \exists u \in \Sigma^*.\, E \to^{I*}_N (u ;
\Gamma)\}$,
$F^E_N \eqdf \{E \mid E \dn\} \cup \{ E_0, E_1 \in Q^E_N \mid E_0 \dn
\con E_1 \dn \}$,
$\Gamma \to^E_N (a; \Gamma') \eqdf \Gamma \Rightarrow^I_N (a;
\Gamma')$
for $\Gamma, \Gamma' \in Q^E_N$. By Cor.~\ref{cor:holes-bounded}, if
$\sem{E}$ has uniform rank at most $N$, then this smaller automaton
accepts $\sem{E}^I$ despite the truncation. If $\sem{E}$ does not have
uniform rank or we choose $N$ smaller than the uniform rank, then the
$N$-truncated automaton recognizes a proper subset of
$\sem{E}^I$. Prop.~\ref{prop:fin-rank-no-urank} gives an example of
this: however we choose $N$, the $N$-truncated automaton fails to
accept the word $a^nb^nca^nb^n$ for $n > N$. This happens because
$\sem{E}$ does not have uniform rank (and that it has rank 2 does not
help).

\subsection{Automaton Finiteness for Regular Expressions with Uniform Rank}

Is the $N$-truncated Antimirov automaton finite? The states $\Gamma$
of $Q^E_N$ are all of length at most $N+1$, so there is hope. The
automaton will be finite if we can find a finite set containing all
the individual regexps $E'$ appearing in the states $\Gamma$. We now
define such a set $E^{\to*}$.

% This will allow us to show that the $N$-truncated
% refined reordering Antimirov automaton of a regexp whose language has
% uniform rank at most $N$ is finite.

\begin{definition} We define functions
  $(\_)^{\squig+}, \RR, (\_)^{\to+}, (\_)^{\to*} : \RE \to \P \RE$ by
\[
\small
\begin{array}{rcl@{\hspace*{3cm}}rcl}
  a^{\squig+}       & \eqdf & \{ 1 \} 
& (E + F)^{\squig+} & \eqdf & E^{\squig+} \cup F^{\squig+} \\
  0^{\squig+}       & \eqdf & \emptyset 
& 1^{\squig+}       & \eqdf & \emptyset \\[1ex]
  (EF)^{\squig+}    & \eqdf & \multicolumn{4}{l}{E^{\squig+} \cup F^{\squig+} \cup E^{\squig+} \cdot \{ F \}
                              \cup \{ E \} \cdot F^{\squig+}
                              \cup E^{\squig+} \cdot F^{\squig+}} \\
  (E^*)^{\squig+}   & \eqdf & \multicolumn{4}{l}{E^{\squig+} \cup \{ E^* \} \cdot E^{\squig+} \cup
                              E^{\squig+} \cdot \{ E^* \} \cup
                              E^{\squig+} \cdot  (\{ E^* \} \cdot
                              E^{\squig+})% \\ 
\cup \ (E^{\squig+} \cdot \{ E^* \}) \cdot
                            E^{\squig+} }
\end{array}
\]
\[
\small
\begin{array}{rcl} %@{\qquad} rcl}
  \RR E & \eqdf & \{ R^I_X E \mid X \subseteq \Sigma \} \\
  E^{\to+}         & \eqdf & \RR (E^{\squig+}) \\
  E^{\to*}         & \eqdf & \{ E \} \cup E^{\to+}
\end{array}
\]
\end{definition}

\begin{proposition} ~ 
\begin{enumerate}
\item  For any $E$, the set $E^{\to*}$ is finite.
%\end{proposition}
% \begin{proof}
%   The set $E^{\squig+}$ can be shown to be finite by induction on
%   $E$. We have that $\RR X$ is finite for any finite $X$.
% \end{proof}

%\begin{lemma}
\item  For any $E$ and $X$, we have $(R^I_X E)^{\to*} \subseteq R^I_X (E^{\to*})$. 
%\end{lemma}
% \begin{proof}
%   $(R^I_s E)^{\squig+} \subseteq R^I_s (E^{\squig+})$ can be shown by
%   induction on $E$. The result follows from the fact that
%   $\RR (R^I_s X) = R^I_s (\RR X)$.
% \end{proof}

%\begin{lemma} \label{lemma:step} 
\item For any $E$, $a$ and $E_l, E_r$, if
  $E \to^I (a; E_l, E_r)$, then $E_l \in R^I_a (E^{\squig+})$ and
  $E_r \in E^{\squig+}$.
%\end{lemma}
% \begin{proof}
%   By induction on the derivation $E \to^I (a; E_l, E_r)$.
% \end{proof}

%\begin{lemma}\label{lemma:squig-step}
\item
  For any $E, E', X, a, E'_l, E'_r$, if $E' \in R^I_X (E^{\squig+})$ and
  $E' \to^I (a; E'_l, E'_r)$, \\ then $E'_l \in R^I_{Xa}(E^{\squig+})$ and
  $E'_r \in R^I_X (E^{\squig+})$.
%\end{lemma}

%\begin{proposition}
\item  For any $E$, $u$ and $E_0, \ldots , E_n$, if $E \to^{I*} (u; E_0, \ldots, E_n)$, then
  $\forall j.\, E_j \in E^{\to*}$.
%\end{proposition}
% \begin{proof}
%   By induction on the derivation $E \to^{I*} (u; \Gamma)$.
% \end{proof}
\end{enumerate}
\end{proposition}

\begin{proposition} \label{prop:urank-refined-antim-finite}
  For every $E$ and $N$, the state set
  $\{ \Gamma \mid \exists u \in \Sigma^*.\, E \to^{I*}_N (u; \Gamma) \}$
  of the $N$-truncated refined Antimirov automaton for $E$ 
  (accepting $\sem{E}^I$ if $\sem{E}$ has uniform rank at most $N$) 
  is finite.
\end{proposition}

\section{Related Work}

Syntactic derivative constructions for regular expressions extended
with constructors for (versions of) the shuffle operation have been
considered, for example, by Sulzmann and Thiemann~\cite{SulzmannT15}
for the Brzozowski derivative and by Broda et al.~\cite{BrodaMMR15}
for the Antimirov derivative. This is relevant to our derivatives
since $L \cdot^I L'$ is by definition a language between $L \cdot L'$
and $L \shuffle L'$. Thus our Brzozowski and Antimirov reordering
derivatives of $EF$ must be between the classical Brzozowski and
Antimirov derivatives of $EF$ and $E \shuffle F$.

% Finite asynchronous automata were introduced by
% Zielonka~\cite{Zielonka87} as a way to characterize recognizable trace
% languages. It is a theorem that a trace language $T$ is recognizable
% if and only if there is a finite asynchronous automaton such that $T$
% is the language accepted by that automaton. Since all recognizable trace
% languages have a star-connected expression defining them, we can also
% construct a finite automaton for every recognizable trace language given its
% star-connected expression. Asynchronous automata allow concurrent
% execution of independent actions but our construction yields a
% traditional automaton.

\section{Conclusion and Future Work}

We have shown that the Brzozowski and Antimirov derivative operations
generalize to trace closures of regular languages in the form of
reordering derivative operations. The sets of Brzozowski resp.\
Antimirov reordering (parts-of-)derivatives of a regexp are generally
infinite, so the deterministic and nondeterministic automata that they
give, accepting the trace closure, are generally infinite. Still, if
the regexp is star-connected, their appropriate quotients are
finite. Also, the set of $N$-bounded refined Antimirov reordering
parts-of-derivatives is finite without quotienting, and we showed
that, if the language of the regexp has uniform rank at most $N$, the
$N$-truncated refined Antimirov automaton accepts the trace
closure. We also proved that star-connected expressions define
languages with finite uniform rank.

In summary, we have established the following picture.

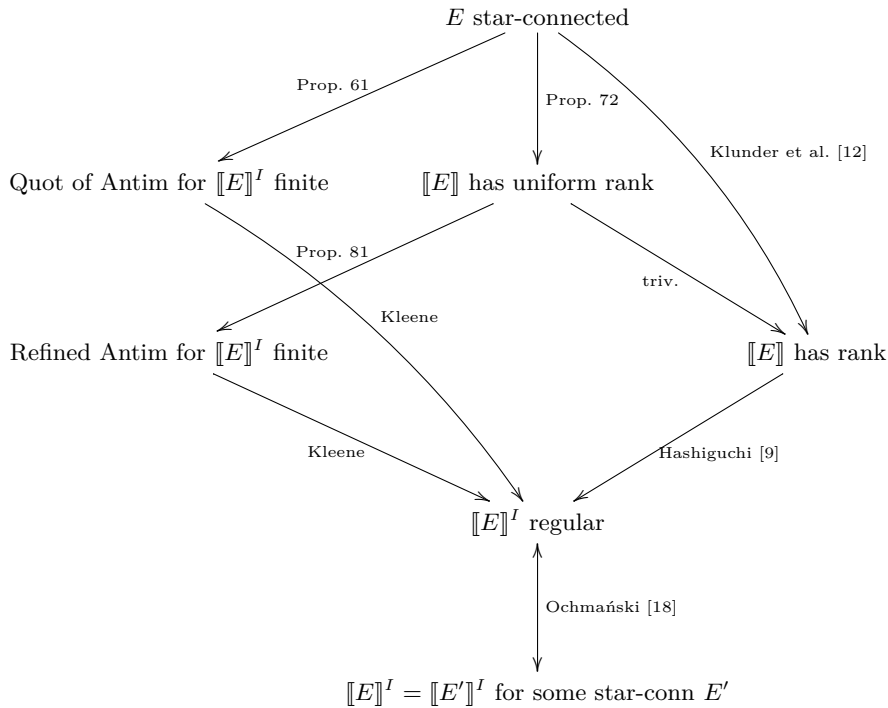
\[
\small
\hspace*{-7mm}
\xymatrix@R=4pc@C=0pc{
%\textrm{$L$ recognizable} \ar@{<->}[r]
%& \textrm{$L$ regular} \ar[d] 
%& \\
& \textrm{$E$ star-connected} \ar[dl]_{\textrm{Prop.~\ref{prop:starconn-antimirov-finite}}} \ar[d]^{\textrm{Prop.~\ref{prop:starconn-urank}}} \ar@/^1pc/[ddr]^{\textrm{Klunder et al.~\cite{KlunderOS05}}}
& \\
\textrm{Quot of Antim for $\sem{E}^I$ finite}  \ar@/^1pc/[ddr]^{\textrm{Kleene}}
& \textrm{$\sem{E}$ has uniform rank} \ar[dl]_{\textrm{Prop.~\ref{prop:urank-refined-antim-finite}}} \ar[dr]_{\textrm{triv.}}
& \\
\textrm{Refined Antim for $\sem{E}^I$ finite} \ar[dr]_{\textrm{Kleene}} 
& & \textrm{$\sem{E}$ has rank} \ar[dl]^{\textrm{Hashiguchi \cite{Hashiguchi91}}} \\
& \textrm{$\sem{E}^I$ regular} \ar@{<->}[d]^{\textrm{Ochma\'nski~\cite{Ochmanski85}}} 
& \\
& \textrm{$\sem{E}^I = \sem{E'}^I$ for some star-conn $E'$} 
&
}
\]

Our intended application of 
reordering derivatives is operational semantics in the
context of relaxed memory (where, e.g., shadow writes, i.e., writes
from local buffers to shared memory, can be reorderable with other
actions). For sequential composition $EF$ it is usually required that,
to execute any action from $F$, execution of $E$ must have
completed. In the jargon of derivatives, this is to say that for an
action from $F$ to become executable, what is left of $E$ has to have
become nullable (i.e., one can consider the execution of $E$
completed). With reordering derivatives, we can execute an action from
$F$ successfully even when what is left of $E$ is not nullable. It
suffices that some sequence of actions to complete the residual of $E$
is reorderable with the selected action of $F$.

In the definitions of the derivative
operations we only use $I$ in one direction, i.e., we do not make use
of its symmetry. It would be interesting to see if our results can
be generalized to the setting of semi-commutations~\cite{ClerboutL87} 
and which changes are required for that.

%%
%% Bibliography
%%

%% Please use bibtex,

%\bibliography{concur19-full}

\end{document}